\documentclass[preprint,11pt]{elsarticle}
\usepackage{amsmath,amsthm,amsfonts,amssymb,bm,graphicx,algorithm,hyperref,cleveref,subcaption,booktabs}
\usepackage[noend]{algorithmic}

\newtheorem{theorem}{Theorem}[section]
\newtheorem{lemma}{Lemma}[section]
\newtheorem{remark}[theorem]{Remark}
\usepackage{xcolor}

\newcommand{\inn}{\mathrm{in}}

\numberwithin{equation}{section}

\newtheorem{assumption}{Assumption}

\crefname{equation}{}{}
\crefname{algorithm}{Algorithm}{}
\crefname{theorem}{Theorem}{}
\crefname{remark}{Remark}{}
\crefname{figure}{Figure}{}

\begin{document}

\title{Efficient Frozen Gaussian Sampling Algorithms for Nonadiabatic Quantum Dynamics at Metal Surfaces}

\author[1]{Zhen Huang}
\ead{hertz@berkeley.edu}

\author[2]{Limin Xu}
\ead{xlm@mails.tsinghua.edu.cn}

\author[3]{Zhennan Zhou\corref{cor1}}
\ead{zhennan@bicmr.pku.edu.cn}

\cortext[cor1]{Corresponding author}

\address[1]{Department of Mathematics, University of California, Berkeley, California 94720 USA}
\address[2]{Department of Mathematical Sciences, Tsinghua University, Beijing 100084, China}
\address[3]{Beijing International Center for Mathematical Research, Peking University, Beijing, 100871, China}

\begin{abstract}
In this article, we propose a Frozen Gaussian Sampling (FGS) algorithm for simulating nonadiabatic quantum dynamics at metal surfaces with a continuous spectrum. This method consists of a Monte-Carlo algorithm for sampling the initial wave packets on the phase space and a surface-hopping type stochastic time propagation scheme for the wave packets. We prove that to reach a certain accuracy threshold, the sample size required is independent of both the semiclassical parameter $\varepsilon$ and the number of metal orbitals $N$, which makes it  one of the most promising methods to study the  nonadiabatic dynamics. The algorithm and its convergence properties are also validated numerically. Furthermore, we carry out numerical experiments including exploring the nuclei dynamics,  electron transfer and  finite-temperature effects, and demonstrate that our method captures the physics which can not be captured by classical surface hopping trajectories.
\end{abstract}
\begin{keyword} metal surfaces \sep Frozen Gaussian Sampling \sep nonadiabatic quantum dynamics \sep Semiclassical Schr\"{o}dinger equation system
\end{keyword}
\maketitle

\section{Introduction}
\label{sec:intro}
Quantum dynamics of atomic systems are well known to exhibit multiscale structures, where  the nuclei are much heavier than the electrons and therefore move much slower. The widely-used Born-Oppenheimer (BO) approximation \cite{born1927quantentheorie} is proposed based on such a  scale-separation structure, where one regards the nuclei as classical particles, and  ignore the couplings between different electronic states. The latter treatment is also termed the adiabatic assumption \cite{born1955dynamical}.

On the one hand, the BO approximation of the Schr\"{o}dinger equation, unfortunately, breaks down dramatically in many scenarios since the adiabatic assumptions are often violated \cite{bunker1977breakdown, pisana2007breakdown, rahinov2011quantifying}. On the other hand, the Schr\"{o}dinger equation, as a high-dimensional model for quantum dynamics, can not be solved in a brute-force way because of the curse of dimensionality. Therefore, how to efficiently calculate dynamics in the non-adiabatic regime is a central topic in  chemistry and material science. 
In order to tackle these issues, various methods based on mixed quantum-classical dynamics \cite{tully1998mixed,kapral1999mixed,kapral2006progress,abedi2014mixed,crespo2018recent} have been developed to simulate non-adiabatic quantum dynamics.

Molecule-metal interfaces are of particular interest since it is related to many different phenomenon in experimental chemistry, such as  chemisorption \cite{newns1969self}, electrochemistry \cite{persson1993applications}, heterogeneous catalysis \cite{luo2016electron} and molecular junctions \cite{nitzan2003electron}. The breakdown of BO approximations at metal surfaces leads to theoretical study of many interesting physical phenomenon: such as  electronic friction \cite{head1995molecular,dou2018perspective}, electron transfer \cite{lindstrom2006photoinduced} and energy transfer \cite{whitmore1982mechanisms}. 
The well-known  Newns-Anderson model \cite{newns1969self}, or the  Anderson-Holstein model \cite{holstein1959studies}, is widely used to model nonadiabatic quantum dynamics at metal surfaces \cite{dou2018perspective,dou2015surface}. In the first quantization, this model could be written as the following Schr\"{o}dinger equation system:
\begin{equation}
\begin{array}{c}{\left\{
\begin{aligned}
\mathrm{i}\varepsilon \partial_{t} u_{0}(t, x)&=-\frac{\varepsilon^{2}}{2} \Delta u_{0}(t, x)+U_{0}(x) u_{0}(t, x)+\varepsilon \sum_{k=1}^Nh V\left({\mathcal{E}_k}, x  \right)u_{k}(t, x),\\
\mathrm{i}\varepsilon \partial_{t} u_{k}(t, x)&=-\frac{\varepsilon^{2}}{2} \Delta u_{k}(t, x)+\left(U_{1}(x)+{\mathcal{E}}_k\right) u_{k}(t, x)+\varepsilon \overline{V}({\mathcal{E}_k}, x) u_{0}(t, x).
\end{aligned}\right.}\\
x\in\mathbb{R}^m, \quad t\in\mathbb{R},\quad k=1,\cdots,N.
\end{array}
\label{eq:discretize_u}
\end{equation}
with initial value
\begin{equation}
\left\{
\begin{aligned}
    u_0(0,x)&=u_0^{\inn}(x),\\ u_k\left(0,x\right)&=0,\qquad k=1,\cdots,N.
    \end{aligned}\right.
    \label{eq:theinitial}
\end{equation}

We will discuss the derivation and details of \cref{eq:discretize_u} in \cref{sec:model}. For now, let us emphasize its major difficulties as of numerical simulations. On the one hand,
there is a non-dimensional parameter $\varepsilon$ presumed to be very small, therefore this model is exactly in the regime of semiclassical dynamics \cite{jin2011mathematical,lasser2020computing}. The solution is highly oscillatory both in space and time \cite{jin2011mathematical}.  On the other hand, we need to solve an $(N+1)\times (N+1)$ matrix Schr\"{o}dinger equation system, where $N$ comes from discretizing the continuum band of metal, and $N$ needs to be large enough in order to capture the physics of continuum band emerging from condensed phases of matter. 
We need a reasonable algorithm whose cost doesn't scale dramatically as $\varepsilon\rightarrow 0$ and $N\rightarrow \infty$.
Methods based on direct discretization of spatial freedom, such as time-splitting spectral methods \cite{bao2002time}, is not ideal since its cost scales  (at least) $O(1/\varepsilon^{1+m})$ ($m$ is the dimension of spatial coordinates) as $\varepsilon\rightarrow 0$ and (at least) $O(N)$ as $N\rightarrow\infty$.

In scientific literature \cite{dou2015surface}, there are generally two kinds of numerical schemes to calculate nonadiabatic quantum dynamics at metal surfaces: Erhenfest dynamics \cite{moss2009ehrenfest} (possibly with electronic friction \cite{head1995molecular})  and surface hopping methods \cite{shenvi2009nonadiabatic}.
With or without electronic friction, Erhenfest dynamics  are in fact a 
 mean-field  approximation of  the coupled electron–nuclear system, therefore  is not valid when strong nonadiabatic dynamics occur \cite{bartels2011energy}.  Independent Electron Surface Hopping (IESH) \cite{shenvi2009nonadiabatic}, as the surface hopping method specifically designed for metal surfaces, is proposed within the same spirit of Tully's original surface hopping method \cite{tully1990molecular}. There are many variants of surface hopping methods in literature, most of which lack a rigorous mathematical foundation and numerical analysis. In \cite{lu2018frozen}, the authors provide a Frozen Gaussian Approximation with surface hopping  (FGA-SH) method along with a rigorous mathematical justification for $O(\varepsilon)$ accuracy. However, its generalization to simulating metal surfaces demands  proper treatment of a large number of bands with a particular  coupling structure,   thus such an investigation poses non-trivial challenges in terms of algorithm design and numerical analysis.
 
In this article, we propose a Frozen Gaussian Sampling (FGS) algorithm for simulating nonadiabatic quantum dynamics at metal surfaces.  Based on a stochastic approximation of the Frozen Gaussian representation of the Schr\"{o}dinger equation system,  the wave packets to be sampled incorporate randomness from both the initial state and from the stochastic evolution. We rigorously prove and numerically validate that the sample size required by the FGS algorithm is independent from both the number of metal orbitals $N$ and the semiclassical parameter $\varepsilon$.
From this sense, FGS as an asymptotic method to the semiclassical Schr\"{o}dinger equation system is superior to the direct application of  other existing asymptotic schemes, such as the WKB methods \cite{engquist2003computational},  Landau-Zener transition asymptotic methods \cite{hagedorn1998landau,jin2011eulerian} and various wave packet based asymptotic methods \cite{jin2008gaussian,lu2010frozen,zhou2018gaussian,miao2021novel}. A by-product is  we are able to prove that the computational cost of  FGS algorithm in finite-band systems \cite{lu2018frozen} is also essentially independent of the semiclassical parameter $\varepsilon$.
What's more, using numerical experiments, we demonstrate  that the FGS method can capture information of observables that are intrinsically quantum, in the sense that such information can not be captured by classical surface hopping trajectories.

The article is organized as followed. In \cref{sec:model}, we give an introduction of the Anderson-Holstein model (Newns-Anderson model) and its nondimensionalization, along with its discretization on the continuous spectrum. In  \cref{sec:algorithm}, we present our algorithm for simulating nonadiabatic dynamics at metal surfaces. After reviewing Frozen Gaussian representation for one-level system in \cref{subsec:FGA}, we derive the  integral representation of the approximate solution based on the surface hopping ansatz in \cref{subsec:FGA_metal}. Based on its stochastic interpretation in \cref{subsec:Stochasticinterpretation}, we present the  Frozen Gaussian Sampling algorithm in \cref{subsec:algorithm}. In \cref{sec:analysis}, we carry out the numerical analysis of our method, indicating that the cost of this stochastic method is independent of both $\varepsilon$ and $N$ (See \cref{thm:main}). Finally, we provide extensive numerical experiments to  illustrate the numerical advantages of our method and some explorations of this model in \cref{sec:numerical}. 

\section{Quantum dynamics of nuclei at metal surfaces}
\label{sec:model}
The quantum dynamics of nuclei at metal surfaces are described by the Newns-Anderson model \cite{newns1969self}, also known as the Anderson-Holstein model \cite{holstein1959studies}, which is a special kind of Anderson impurity model: the electronic ground state of the molecule is treated as the system orbital, while the metal electronic orbitals are treated as bath orbitals. The energies of metal electronic orbitals form a continuous spectrum $[\mathcal{E}_a,\mathcal{E}_b]$.
The model Hamiltonian is usually written in the following second quantization form:
\begin{equation}
\begin{aligned}
    \hat H &= \frac{\hat p^2}{2m_{\text{n}}}+ U_1(\hat x) + h(\hat x)\hat{d}^{\dagger}\hat{d} \\
   & +\int_{\mathcal{E}_a}^{\mathcal{E}_b}(\mathcal{E}-\mu) \hat c_{\mathcal{E}}^{\dagger}\hat c_{\mathcal{E}}\mathrm{d}\mathcal{E}
    +\int_{\mathcal{E}_a}^{\mathcal{E}_b}\left(V(\mathcal{E},\hat{x}) \hat c_{\mathcal{E}}^{\dagger}\hat d+\overline{V}(\mathcal{E},\hat{x})\hat d^{\dagger}\hat c_{\mathcal{E}}\right)\mathrm{d}\mathcal{E}.
\end{aligned}
\end{equation}
Here $\hat p$ is the momentum operator, $\hat x$ is the position operator, $m_{\text{n}}$ is the mass of the nuclei, $U_1(\hat{x})$ is the
nuclear potential for the neutral molecule, $U_1(\hat{x})+h(\hat x)$ is the
nuclear potential for the charged molecule, $\mu$ is the chemical potential. 
And,  $\hat{d}$ and $\hat{d}^{\dagger}$ are the annihilation and creation operators for the electronic ground state of the molecule, $\hat{c}_{\mathcal E}$ and $\hat{c}_{\mathcal E}^{\dagger}$ are the annihilation and creation operators for metal electronic orbitals with energy level $\mathcal E\in[\mathcal E_a,\mathcal E_b]$, $V\left(\mathcal{E},x\right)$ describes the coupling between the molecule and metal orbitals, and  $\overline{V}\left(\mathcal{E},x\right)$ means the complex conjugate of ${V}\left(\mathcal{E},x\right)$.

Now let us demonstrate the non-dimensionalization of this model.  
Let $\ell$ be the characteristic length scale and $E$ be the characteristic energy scale, we  introduce the dimensionless parameter $\varepsilon$, which is referred to as the semiclassical parameter:
\begin{equation}
    \varepsilon=\frac{\hbar}{\ell\sqrt{mE}}.
\end{equation}
Let $E_V$ be the characteristic  scale of molecule-metal coupling energy, we are ready to perform the following rescaling (similar as in \cite{cao2017lindblad}):
\begin{equation}
   \begin{array}{c}
    x=\ell \tilde{x},\quad t=\sqrt{\frac{m\ell^2}{E}}\tilde{t},\quad 
    \mathcal{E}=E\tilde{\mathcal{E}},\quad 
    \mu=E\tilde{\mu}, \\
    \tilde U_1(\tilde x) =\frac{U_1(x)}{E},\quad
     \tilde h(\tilde x) =\frac{h(x)}{E},\quad,\tilde V(\tilde{\mathcal{E}},
    \tilde x)=\frac{V(\mathcal{E},x)}{E_V},\quad \hat H=E{{\hat H_{\text{non}}}}.
   \end{array}
\end{equation}
And particularly,  we  choose the   regime
$     E_V=\varepsilon E
$, and the corresponding non-dimensionalized Hamiltonian is
\begin{equation}
\begin{aligned}
    \hat H_{\text{non}} &=-\frac{\varepsilon^{2}}{2}  \nabla_{\tilde{x}}^{2}+ \tilde U_1(\tilde x) + \tilde h(\tilde x)\hat{d}^{\dagger}\hat{d} \\
   & +\int_{\tilde{\mathcal{E}}_a}^{\tilde{\mathcal{E}}_b}(\tilde{\mathcal{E}}-\tilde\mu) \hat c_{\tilde{\mathcal{E}}}^{\dagger}\hat c_{\tilde{\mathcal{E}}}\mathrm{d}\mathcal{E}
    +\varepsilon\int_{\tilde{\mathcal{E}}_a}^{\tilde{\mathcal{E}}_b}\left(\tilde V(\tilde{\mathcal{E}},\tilde{x}) \hat c_{\mathcal{E}}^{\dagger}\hat d+\overline{V}(\mathcal{E},\hat{x})\hat d^{\dagger}\hat c_{\mathcal{E}}\right)\mathrm{d}\mathcal{E}.
\end{aligned}
\end{equation}    

We can see that the coupling between the molecule and metal electronic orbitals is chosen to be $O(\varepsilon)$. By the scientific literature \cite{lindstrom2006photoinduced,dou2018perspective}, this scenario can be intepreted as the  weak-coupling regime of nonadiabatic dynamics at metal surfaces. However, even though the coupling is considered to be weak, since the Hamiltonian has a multiscale structure, such an $O(\varepsilon)$ coupling  still causes an $O(1)$ change in the density of states. 
As a matter of fact, in the adiabatic representation, the coupling between different levels is also formally  $O(\varepsilon)$, see \cite{lu2018frozen}.

 In this article,  our starting point is to rewrite $\hat H_{\operatorname{non}}$ in the first quantization.
Let  $x\in\mathbb{R}^m$ represents the nuclei degrees of freedom, and let $U_0(x)=U_1(x)+h(x)$ denote the nuclear potential for the charged molecule, and recall that $U_1(x)$ represents the nuclear potential for the neutral molecule. The above model is described by the following Schr\"{o}dinger equation system:

\begin{equation}\left\{
\begin{aligned}
\mathrm{i}\varepsilon \partial_{t} \psi_{0}(t, x)&=-\frac{\varepsilon^{2}}{2} \Delta \psi_{0}(t, x)+U_{0}(x) \psi_{0}(t, x)+\varepsilon \int_{\mathcal{E}_a}^{{\mathcal{E}_b}} \mathrm d \tilde{\mathcal{E}} V(\tilde{\mathcal{E}}, x) \psi_{1}(t,  \tilde{\mathcal{E}},x), \\
\mathrm{i}\varepsilon \partial_{t} \psi_{1}(t, \mathcal{E},x)&=-\frac{\varepsilon^{2}}{2} \Delta \psi_{1}(t, \mathcal{E},x)+\left(U_{1}(x)+{\mathcal{E}}\right) \psi_{1}(t, \mathcal{E},x)+\varepsilon \overline{V}\left({\mathcal{E}}, x\right) \psi_{0}(t, x).
\end{aligned}\right.
 \label{eq:continuous_u}
\end{equation}
We have dropped  the “tilde”  and without loss of generality we assume $\mu=0$ for simplification.
Here $\psi_0(t,x)\in L^2(\mathbb{R}^m)$ represents the nuclei wave function when the independent electron propagates into the molecule and therefore the molecule is charged, while $\psi_1\left(t,\mathcal{E},x\right)\in L^2\left([\mathcal{E}_a,\mathcal{E}_b]\times\mathbb{R}^m\right)$ represents the nuclei wave function when the molecule is neutral and the electron is in the metal orbital with energy $\mathcal E$.  $\psi_0(t,x)$ and $\psi_1\left(t,\mathcal{E},x\right)$ satisfy the following normalization condition (also known as, mass conservation):
\begin{equation}
    1=m(t)=\int_{\mathbb{R}^m}\left(|\psi_0(t,x)|^2\mathrm{d}x+ \int_{\mathcal{E}_a}^{\mathcal{E}_b}|\psi_1\left(t,\mathcal{E},x\right)|^2\mathrm{d}\mathcal{E}\right)\mathrm{d}x,\quad \forall t\geq 0.
\end{equation}

We remark that in this work we are considering a {\emph{closed}} quantum system. This is exactly the viewpoint of Newn's pioneering work \cite{newns1969self}. In some scientific literature \cite{jin2021nonadiabatic,dou2015surface}, the setup of metal surfaces are open quantum systems, where one focus only on the molecule orbital (the orbital corresponding to $\psi_0$) and treat the metal orbitals (the orbitals corresponding to $\psi_1$) as heat baths with Boltzmann distribution.   We would like to emphasize  that even in the closed-system treatment, the nonadiabatic features of quantum dynamics   have already arisen. In other words, one do not need the open quantum system and canonical ensemble assumption to encounter nonadiabatic quantum dynamics at metal surfaces.  We  leave the open quantum system treatment of metal surfaces to future work.

The energy $E(t)$ of the above system \cref{eq:continuous_u} is defined as:
\begin{equation}
\begin{aligned}
E(t)&=\int_{\mathbb{R}^m}\left(\frac{\varepsilon^2}{2}\left|\nabla_x\psi_0(t,x)\right|^2+U_0(x)\left|\psi_0(t,x)\right|^2\right)\mathrm{d}x\\
&+\int_{\mathbb{R}^m}\int_{\mathcal{E}_a}^{\mathcal{E}_b}\left(
\frac{\varepsilon^2}{2}\left|\nabla_x\psi_1\left(t,\mathcal{E},x\right)\right|^2
+\left(U_1(x)+\mathcal{E}\right)\left|\psi_1\left(t,\mathcal{E},x\right)\right|^2\right)\mathrm{d}\mathcal{E}\mathrm{d}x\\
&+\varepsilon\int_{\mathbb{R}^m}\int_{\mathcal{E}_a}^{\mathcal{E}_b}\left(
\bar\psi_0(t,x) V\left(\mathcal{E},x\right) \psi_{1}\left(t,\mathcal{E},x\right) +\psi_0(t,x) \overline{V}\left(\mathcal{E},x\right) \bar\psi_{1}\left(t,\mathcal{E},x\right)
\right)\mathrm{d}\mathcal{E}\mathrm{d}x.
\end{aligned}
\end{equation}
It's standard to verify that \cref{eq:continuous_u} satisfies the energy conservation: $E(t)=E(0)$ for any $t\geq 0$.

For simplicity, we assume at the initial state the molecule is charged, and thus the electron is in the adsorbate orbital. In other words, we consider the following initial condition:
\begin{equation}
    \psi_0(0,x)=\psi_0^{\inn}(x),\quad \psi_1\left(0,\mathcal{E},x\right)=0.
\end{equation}
However, since \cref{eq:continuous_u} is a linear equation system, the results in this article are applicable to any initial conditions.

A typical treatment of \cref{eq:continuous_u}  is to discretize the electronic continuum spectrum $[{\mathcal{E}_a},{\mathcal{E}_b}]$ into $N$ separate orbitals.  Then we obtain the discrete Anderson-Holstein model, which is a commonly used model
to investigate surface hopping at molecule-metal surfaces in scientific literature (for example, see \cite{shenvi2009nonadiabatic,jin2021nonadiabatic,dou2015surface}).
For the sake of simplicity, we use the following equidistant discretization throughout this article:
\begin{equation}
    \mathcal{E}_k=\mathcal{E}_a+(i-\frac{1}{2})h,\quad h=\frac{1}{N}(\mathcal{E}_b-\mathcal{E}_a),\quad k=1,\cdots,N.
\end{equation}
Note that the generalization to any non-equidistant discretization is trivial. Let $$u_0(t,x)=\psi_0(t,x), \quad u_k(t,x)=\psi_1(t,\mathcal{E}_k,x),\quad k=1,\cdots,N.$$
Then we can write down the Anderson-Holstein model for $\{u_k(t,x)\}_{k=0}^N$, which we have already presented in \cref{sec:intro}:
\begin{equation*}
\begin{array}{c}{\left\{
\begin{aligned}
\mathrm{i}\varepsilon \partial_{t} u_{0}(t, x)&=-\frac{\varepsilon^{2}}{2} \Delta u_{0}(t, x)+U_{0}(x) u_{0}(t, x)+\varepsilon \sum_{k=1}^Nh V\left({\mathcal{E}_k}, x  \right)u_{k}(t, x),\\
\mathrm{i}\varepsilon \partial_{t} u_{k}(t, x)&=-\frac{\varepsilon^{2}}{2} \Delta u_{k}(t, x)+\left(U_{1}(x)+{\mathcal{E}}_k\right) u_{k}(t, x)+\varepsilon \overline{V}({\mathcal{E}_k}, x) u_{0}(t, x).
\end{aligned}\right.}\\
x\in\mathbb{R}^m, \quad t\in\mathbb{R},\quad k=1,\cdots,N.
\end{array}
\end{equation*}
with initial value
\begin{equation*}
    u_0(0,x)=u_0^{\inn}(x),\quad u_k\left(0,x\right)=0,\quad k=1,\cdots,N.
\end{equation*}
Then for the systems of $\{u_k(t,x)\}_{k=0}^N$, the total mass is defined as
\begin{equation}
    m_N(t)=\int_{\mathbb{R}^m}\left(|u_0(t,x)|^2\mathrm{d}x+ h\sum_{k=1}^N|u_k(t,x)|^2\right)\mathrm{d}x.
\end{equation}
Here the index $N$ in $m_N(t)$ represents that the metal continuum in this model is discretized into $N$ orbitals. Similarly, we can define the system energy $E_N(t)$:
\begin{equation}
\begin{aligned}
E_N(t)&=\int_{\mathbb{R}^m}\frac{\varepsilon^2}{2}\left|\nabla_x u_0\right|^2+U_0(x)\left|u_0\right|^2+ h\sum_{k=1}^N\left(\frac{\varepsilon^2}{2}|\nabla_x u_k|^2+\left(U_1(x)+\mathcal{E}\right)\left|u_k\right|^2\right) \mathrm{d}x  \\
&+\varepsilon\int_{\mathbb{R}^m}h\sum_{k=1}^N\left(V(\mathcal{E}_k,x)\bar u_0(t,x)  u_{k}(t,x) +\overline{V}(\mathcal{E}_k,x)u_0(t,x) \bar u_{k}(t,x)
\right)\mathrm{d}x.
\end{aligned}
\end{equation}
It's also straightforward to verify that solutions to \cref{eq:discretize_u} satisfy the conservation of mass and energy:
\begin{equation}
    m_N(t)=m_N(0),\quad E_N(t)=E_N(0),\quad \forall t\geq 0. \label{eq:discretizeconsevation}
\end{equation}

  Dynamics \cref{eq:discretize_u}
could be regarded as a  $(N+1)-$level semiclassical Schr\"{o}dinger equation system.
As we have mentioned in \cref{sec:intro}, 
our goal is to design an algorithm for \cref{eq:discretize_u},
whose computational cost does not essentially depend on both the semiclassical parameter $\varepsilon$ and the number of metal orbitals $N$.

\section{Algorithm}
\label{sec:algorithm}
In this section, we introduce the Frozen Gaussian Sampling (FGS) algorithm for the system \cref{eq:discretize_u}.
\subsection{Review of Frozen Gaussian  representation for one-level system}
\label{subsec:FGA}
Before dealing with the multi-level system \cref{eq:discretize_u}, we first give a brief review of the Frozen Gaussian  representation for a one-level system.
Consider the following semiclassical Schr\"{o}dinger equation:
\begin{equation}
i \varepsilon \frac{\partial}{\partial t} u(t, x)=-\frac{\varepsilon^{2}}{2} \Delta u(t, x)+U(x) u(t, x),\quad u(0,x)=u_{\mathrm{in}}(x).
\label{eq:onelevelSchrodinger}
\end{equation}
It is known that the Frozen Gaussian approximation of \cref{eq:onelevelSchrodinger} is an approximation derived from asymptotic analysis  with $O(\varepsilon)$ error \cite{swart2009mathematical}. It is based on an integral representation of Frozen Gaussian wave packets on phase space:
\begin{equation}
u_{\mathrm{FG}}(t, x)=\frac{1}{(2 \pi \varepsilon)^{\frac{3 m}{2}}} \int_{\mathbb{R}^{2 m}} A(t, q, p) \mathrm{e}^{\frac{\mathrm{i} \Theta(t, x, q, p)}{\varepsilon}} \mathrm{d} q \mathrm{~d} p.
\label{eq:onelevelFGA}
\end{equation}
where the phase function $\Theta$ is
\begin{equation}
\Theta(t, x, q, p)=S(t, q, p)+P(t, q, p) \cdot(x-Q(t, q, p))+\frac{\mathrm{i}}{2}|x-Q(t, q, p)|^{2}.
\end{equation}
The evolution of $(Q(t,q,p),P(t,q,p))$ is governed by the Hamiltonian flow with $(q,p)$ as its initial value:
\begin{equation}
\begin{aligned}
\frac{\mathrm{d}}{\mathrm{d} t} Q&=\partial_{p} h(Q, P),\quad Q(0,q,p)=q. \\
\frac{\mathrm{d}}{\mathrm{d} t} P&=-\partial_{q} h(Q, P),\quad P(0,q,p)=p.
\end{aligned}
\label{eq:Hamiltoniandynamics}
\end{equation}
where $h(q, p)=\frac{1}{2}|p|^{2}+U(q)$ is the classical Hamiltonian.
The dynamics of the action $S(t,q,p)$ and amplitude $A(t,q,p)$ are
\begin{equation}
\frac{\mathrm{d}}{\mathrm{d} t} S(t, q, p)=\frac{1}{2}|P|^{2}-U(Q),\quad S(0, q, p)=0.
\end{equation}
\begin{equation}
    \begin{aligned}
    \frac{\mathrm{d}}{\mathrm{d} t} A(t, q, p)&=\frac{A}{2} \operatorname{tr}\left(Z^{-1}\left(\partial_{z} P-\mathrm{i} \partial_{z} Q \nabla_{Q}^{2} U(Q)\right)\right),\\
    A(0, q, p)&=2^{\frac{m}{2}} \int_{\mathbb{R}^{m}} u_{\mathrm{in}}(y) e^{\frac{\mathrm i}{\varepsilon}\left(-p \cdot(y-q)+\frac{i}{2}|y-q|^{2}\right)} \mathrm{d} y.
    \end{aligned}\label{eq:dynamicsforamplitude}
\end{equation}
Here we use the short hand notation:
\begin{equation}
\partial_{z}=\partial_{q}-i \partial_{p} \quad \text { and } \quad Z=\partial_{z}(Q+i P).
\end{equation}
We emphasize here that  the initial condition $u_{\text{in}}$ has already been encoded into $A(0,q,p)$.

The above Frozen Gaussian dynamics are derived using asymptotic matching \cite{swart2009mathematical,lu2010frozen}. The Frozen Gaussian representation for the initial value $u_{\inn}$ has the following property, which we quote directly from \cite{swart2009mathematical}:
\begin{theorem}
For $u \in L^{2}\left(\mathbb{R}^{m}\right)$, we have the following decomposition
\begin{equation}
u_{\mathrm{in}}(x)=\frac{1}{(2 \pi \varepsilon)^{\frac{3 m}{2}}} \int_{\mathbb{R}^{2 m}} A(0, q, p) \mathrm{e}^{\frac{\mathrm{i} \Theta(0, x, q, p)}{\varepsilon}} \mathrm{d} q \mathrm{~d} p,
\end{equation}
where $A(0,q,p)$ is given as in \cref{eq:dynamicsforamplitude}.
\end{theorem}

And, for the Frozen Gaussian representation \cref{eq:onelevelFGA} of the one-level semiclassical  dynamics \cref{eq:onelevelSchrodinger}, we have the following error estimate \cite{swart2009mathematical,lu2010frozen}:
\begin{theorem}
Let $u(t, x)$ be the exact solution for the one-level semiclassical dynamics \cref{eq:onelevelSchrodinger}. Assume that $U(x) \in C^{\infty}\left(\mathbb{R}^{m}\right)$. For any $t>0$, we have the following $L^{2}$ error estimate for the Frozen Gaussian representation \cref{eq:onelevelFGA}
\begin{equation}
\left\|u_{\mathrm{FG}}(t, \cdot)-u(t, \cdot)\right\|_{L^{2}} \leq C_{m}(t) \varepsilon\left\|u_{\mathrm{in}}(\cdot)\right\|_{L^{2}},
\end{equation}
where $C_m(t)$ is a constant independent of $\varepsilon$.
\end{theorem}
\subsection{Integral representation  based on the surface hopping ansatz}
\label{subsec:FGA_metal}
The integral representation of the solution of \cref{eq:discretize_u} is based on the following surface hopping ansatz: (as illustrated in \cref{fig:hopping})
\begin{itemize}
    \item A wave packet starts from the molecule orbital.
    \item While the wave packet is on the molecule orbital, it might hop onto one of the metal surfaces; while the wave packet is on one of the metal surfaces, it might hop back onto the molecule orbital. No hopping between different metal surfaces is allowed. 
    \item When the wave packet hasn't hopped out of the current energy surface, it would propagate subject to the dynamics specified by the current energy surface.
\end{itemize}

\begin{figure}[h]
    \centering
    \includegraphics[width=100mm]{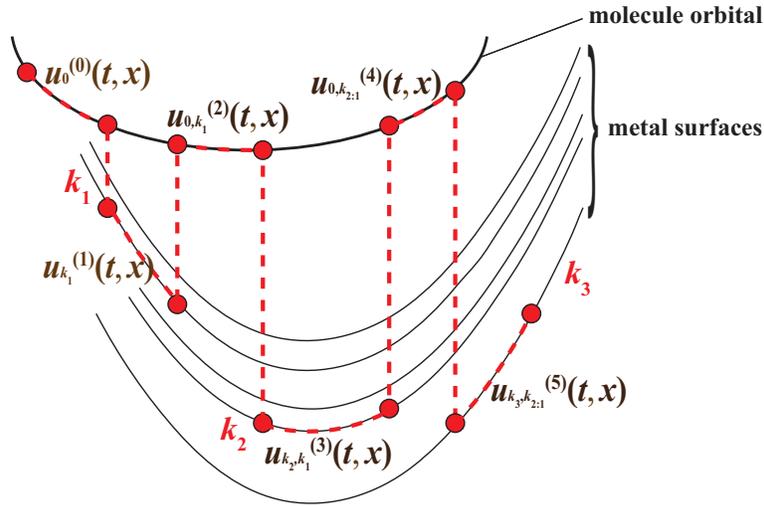}
    \caption{Illustration of the surface hopping ansatz: $k_1,k_2,k_3\in \{1,2,\cdots,N\}$ represents the indices of the metal orbitals that the wave packet has hopped onto. $u_0^{(0)}$ represents the wave packet that hasn't ever hopped; $u_{k_1}^{(1)}$ represents the wave packet that has hopped once and is on the $k_1$-th metal surface;  $u_{0, k_1}^{(2)}$ represents the wave packet that has hopped twice and was on the $k_1$-th metal surface before the second jump; $u_{k_2, k_1}^{(3)}$ represents the wave packet that has hopped three times, is currently  on the $k_2$-th metal surface  and was on the $k_1$-th metal surfaces in history. The meaning of $u_{0, k_{2:1}}^{(4)}$ and $u_{k_3, k_{2:1}}^{(5)}$  could be determined similarly.
    }
    \label{fig:hopping}
\end{figure}

The hopping probability and the detailed wave packet dynamics will be specified later.
For now, let us introduce some notations. Let $k_1\in \{1,2,\cdots,N\}$ be the index of the metal orbital that the wave packet is on when it is the first time for the wave packet to ever hop onto the metal surfaces. In the same way, we can define $k_2,k_3,\cdots$.
The wave packet on the molecule orbital is denoted as 
$$
u^{(2n)}_{0,k_{n:1}},\quad n=0,1,2,\cdots,\quad k_1,k_2,\cdots ,k_n\in \{1,2,\cdots,N\},
$$
where the superscript $2n$ indicates the times that the wave packet has already jumped (which must be an even integer since the wave packet has to hop out and  back to be on the molecule orbital), the subscript $0$ means that  the wave packet is currently on the molecule orbital, and  $k_{n:1}$  is the short hand notation for $k_1,k_2,\cdots,k_n$, which documents the indices of  all metal orbitals that this wave packet has ever jumped onto.
Similarly, the wave packet on the metal surfaces are denoted as
$$
u^{(2n+1)}_{k,k_{n:1}}, \quad n=0,1,2,\cdots,\quad k,k_1,k_2,\cdots ,k_n\in \{1,2,\cdots,N\},
$$
where the superscript is still the hopping times (an odd integer) and the subscript is the current metal orbital index $k$ and the historical metal orbitals indices $k_{n:1}$.
The wave packets could be written in the following unified formulation:
$$ u_{k,k_{[\nu/2]:1}}^{(\nu)},\quad  \nu = 0,1,2,\cdots,\quad k_1,k_2,\cdots ,k_n\in \{1,2,\cdots,N\}, $$
where $k=0$ when $\nu$ is even, which means the wave packet is on the molecule orbital, and $k\in\{1,\cdots,N\}$ when $\nu$ is odd, which means that the wave packet is on the $k$-th metal surface.
Later, when we introduce other quantities of the wave packet, we will use the same index system.

Now we are ready to present the surface hopping ansatz of the wavefunction:
\begin{equation}
\begin{aligned}
{u}_{0,\text{FG}}(t,x)&=  u_0^{(0)}(t,x)+ \sum_{k_1=1}^N u_{0,k_1}^{(2)}(t,x)+ \sum_{k_1,k_2=1}^N u_{0,k_{2:1}}^{(4)}(t,x)+\cdots,\\
{u}_{k,\text{FG}}(t,x)&=  u_{k}^{(1)}(t,x)+ \sum_{k_1=1}^N u_{k,k_1}^{(3)}(t,x)+ \sum_{k_1,k_2=1}^N u_{k,k_{2:1}}^{(5)}(t,x)+\cdots, k=1:N.
\end{aligned}
\end{equation}

Since $ u_0^{(0)}(t,x)$ represents the wave packet that only propagates on $U_0$ and hasn't switched to the metal surfaces, it is thus given by the same ansatz as in one-level system:
\begin{equation}
 u_0^{(0)}(t, x)=\frac{1}{(2 \pi \varepsilon)^{3 m / 2}} \int A_0^{(0)}\left(t, z_{0}\right) \exp \left(\frac{\mathrm i}{\varepsilon} \Theta_0^{(0)}\left(t, z_{0}, x\right)\right) \mathrm{d} z_{0}.
\end{equation}
Here $z_0=(q_0,p_0)$, and
\begin{equation}
\Theta_0^{(0)}(t, x, q, p)=S_0^{(0)}(t, q, p)+P_0^{(0)}(t, q, p) \cdot(x-Q_0^{(0)}(t, q, p))+\frac{\mathrm{i}}{2}|x-Q_0^{(0)}(t, q, p)|^{2},
\end{equation}
and the dynamics of $Q_0^{(0)}$, $P_0^{(0)}$, $A_0^{(0)}$, $S_0^{(0)}$ will be specified below.

For $\nu>0$, let us denote $T_{\nu: 1}=\left(t_{\nu}, \ldots, t_{1}\right)$ as a sequence of times in decreasing order:
$$
t_0=0 \leq t_{1} \leq t_{2} \leq \cdots \leq t_{\nu} \leq t=t_{\nu+1}.
$$
And now we are ready to introduce our ansatz for $\nu=2n$ and $\nu=2n+1$:
\begin{equation}
\begin{aligned}
 u_{0,k_{n:1}}^{(2n)}(t, x)=  \frac{1}{(2 \pi \varepsilon)^{3 m / 2}} &\int \mathrm{d} z_{0} \int_{0}^{t} \mathrm{~d} t_{2n} \int_{0}^{t_{2n}} \mathrm{~d} t_{2n-1} \cdots\int_{0}^{t_{2}} \mathrm{d} t_{1}\\
  &\left(\prod_{j=1}^{n}\left(h\tau_{k_j,k_{j-1:1}}^{(2j-1)}\left(T_{2j-1: 1}, z_{0}\right)\tau_{{0,k_{j:1}}}^{(2j)}\left(T_{2j: 1}, z_{0}\right)\right) \right)
  \\&\times A_{0,k_{n:1}}^{(2n)}\left(t, T_{2n: 1}, z_{0}\right) \exp \left(\frac{\mathrm i}{\varepsilon} \Theta_{0,k_{n:1}}^{(2n)}\left(t, T_{2n: 1}, z_{0}, x\right)
\right).
\end{aligned}
\end{equation}

\begin{equation}
\begin{aligned}
 u_{k,k_{n:1}}^{(2n+1)}(t, x)=& \frac{1}{(2 \pi \varepsilon)^{3 m / 2}}  \int \mathrm{d} z_{0} \int_{0}^{t} \mathrm{~d} t_{2n+1} \int_{0}^{t_{2n}} \mathrm{~d} t_{2n} \cdots\int_{0}^{t_{2}} \mathrm{d} t_{1}\\
   &\tau_{k,k_{n:1}}^{(2n+1)}\left(T_{2n+1: 1}, z_{0}\right)\times \left(\prod_{j=1}^{n}\left(h\tau_{k_j,k_{j-1:1}}^{(2j-1)}\left(T_{2j-1: 1}, z_{0}\right)\tau_{0,k_{j:1}}^{(2j)}\left(T_{2j: 1}, z_{0}\right)\right) \right)\\
&\times A_{k,k_{n:1}}^{(2n+1)}\left(t, T_{2n+1: 1}, z_{0}\right) \exp \left(\frac{\mathrm i}{\varepsilon} \Theta_{k,k_{n:1}}^{(2n+1)}\left(t, T_{2n+1: 1}, z_{0}, x\right)
\right).
\end{aligned}
\end{equation}
Here $h=(\mathcal{E}_b-\mathcal{E}_a)/N$.  
For both $\nu=2n$ and $\nu = 2n+1$, we have
\begin{equation}
\begin{aligned}
\Theta_{k,k_{[{\nu}/{2}]:1}}^{(\nu)}\left(t, T_{\nu: 1}, z_{0}, x\right)=& S_{k,k_{[{\nu}/{2}]:1}}^{(\nu)} \left(t, T_{\nu: 1}, z_{0}\right)+\frac{\mathrm i}{2}\left|x-Q_{k,k_{[{\nu}/{2}]:1}}^{(\nu)}\left(t, T_{\nu: 1}, z_{0}\right)\right|^{2}\\+&P_{k,k_{[{\nu}/{2}]:1}}^{(\nu)}\left(t, T_{\nu: 1}, z_{0}\right) \cdot\left(x-Q_{k,k_{[{\nu}/{2}]:1}}^{(\nu)}\left(t, T_{\nu: 1}, z_{0}\right)\right).
\end{aligned}
\label{eq:Thetarelation}
\end{equation}
Compared with the Frozen Gaussian representation for one-level system \cref{eq:onelevelFGA}, the above ansatz includes the integration over all possible jumping times, while the parameters $\tau_{k,k_{[\nu/2]:1}}^{(\nu)}$ accounting for these jumps are to be determined.  

Based on the ansatz, we can determine the dynamics of $S_{k,k_{[{\nu}/{2}]:1}}^{(\nu)},P_{k,k_{[{\nu}/{2}]:1}}^{(\nu)}$, $Q_{k,k_{[{\nu}/{2}]:1}}^{(\nu)}$, $A_{k,k_{[{\nu}/{2}]:1}}^{(\nu)}$, and the value of $\tau_{k,k_{[\nu/2]:1}}^{(\nu)}$.
The detailed derivation is based on the asymptotic matching, and is stated in \ref{sec:derivation}.
The dynamics of $S_{k,k_{[{\nu}/{2}]:1}}^{(\nu)},P_{k,k_{[{\nu}/{2}]:1}}^{(\nu)}, Q_{k,k_{[{\nu}/{2}]:1}}^{(\nu)}$ and $ A_{k,k_{[{\nu}/{2}]:1}}^{(\nu)}$  for $t\in [t_{\nu},t_{\nu+1}]$ are
\begin{equation}
\begin{aligned}
\frac{\mathrm{d}}{\mathrm{d} t} Q_{k,k_{[{\nu}/{2}]:1}}^{(\nu)}=& P_{k,k_{[{\nu}/{2}]:1}}^{(\nu)}, \\
\frac{\mathrm{d}}{\mathrm{d} t} P_{k,k_{[{\nu}/{2}]:1}}^{(\nu)}=&-\nabla_Q \widetilde U\left(t,Q_{k,k_{[{\nu}/{2}]:1}}^{(\nu)}\right), \\
\frac{\mathrm{d}}{\mathrm{d} t} S_{k,k_{[{\nu}/{2}]:1}}^{(\nu)}=& \frac{1}{2}\left(P_{k,k_{[{\nu}/{2}]:1}}^{(\nu)}\right)^{2}-\widetilde U\left(t,Q_{k,k_{[{\nu}/{2}]:1}}^{(\nu)}\right), \\
\frac{\mathrm{d}}{\mathrm{d} t} A_{k,k_{[{\nu}/{2}]:1}}^{(\nu)}=& \frac{1}{2} A_{k,k_{[{\nu}/{2}]:1}}^{(\nu)} \operatorname{tr}\left(\left(Z^{(\nu)}\right)^{-1}\left(\partial_{z} P_{k,k_{[{\nu}/{2}]:1}}^{(\nu)}-i \partial_{z} Q_{k,k_{[{\nu}/{2}]:1}}^{(\nu)} \nabla_{Q}^{2} \widetilde U\left(t,Q_{k,k_{[{\nu}/{2}]:1}}^{(\nu)}\right)\right)\right),
\end{aligned}
\label{eq:groupdynamics}
\end{equation}
where 
\begin{equation}
    \widetilde{U}(t,Q)=\left\{
    \begin{aligned}
    U_0(Q),\quad &t\in [t_{2j}, t_{2j+1}]
    \\
    U_1(Q)+\mathcal{E}_{k_j},\quad &t\in [t_{2j-1}, t_{2j}]
    \end{aligned}
    \right. ,
\end{equation}
and the initial value of time interval $[0,t_1]$ is
\begin{equation}
    Q_0^{(0)}(0,z_0)=q,\quad P_0^{(0)}(0,z_0)=q,\quad S_0^{(0)}(0,z_0)=0,
    \label{eq:initialQPS}
\end{equation}
\begin{equation}
     A_0^{(0)}(0, z_0)=2^{\frac{m}{2}} \int_{\mathbb{R}^{m}} u_{0,\mathrm{in}}(y) e^{\frac{\mathrm i}{\varepsilon}\left(-p \cdot(y-q)+\frac{i}{2}|y-q|^{2}\right)} \mathrm{d} y.
     \label{eq:dynamicsforamplitude2}
\end{equation}
and the initial value on time interval $[t_{\nu},t_{\nu+1}]$ is chosen to makes sure all of them are continuous functions.
\begin{equation}
\begin{aligned}
Q_{k,k_{\lceil{(\nu+1)}/{2}\rceil:1}}^{(\nu+1)} \left(t_{\nu+1}, T_{\nu+1: 1}, z_{0}\right)&=
    Q_{k,k_{[{\nu}/{2}]:1}}^{(\nu)} \left(t_{\nu+1}, T_{\nu: 1}, z_{0}\right),\\P_{k,k_{\lceil{(\nu+1)}/{2}\rceil:1}}^{(\nu+1)} \left(t_{\nu+1}, T_{\nu+1: 1}, z_{0}\right)&=
    P_{k,k_{[{\nu}/{2}]:1}}^{(\nu)} \left(t_{\nu+1}, T_{\nu: 1}, z_{0}\right),\\
    S_{k,k_{\lceil{(\nu+1)}/{2}\rceil:1}}^{(\nu+1)} \left(t_{\nu+1}, T_{\nu+1: 1}, z_{0}\right)&=
    S_{k,k_{[{\nu}/{2}]:1}}^{(\nu)} \left(t_{\nu+1}, T_{\nu: 1}, z_{0}\right),\\
A_{k,k_{\lceil{(\nu+1)}/{2}\rceil:1}}^{(\nu+1)} \left(t_{\nu+1}, T_{\nu+1: 1}, z_{0}\right)&=
    A_{k,k_{[{\nu}/{2}]:1}}^{(\nu)} \left(t_{\nu+1}, T_{\nu: 1}, z_{0}\right).
\end{aligned}
\label{eq:initialQPSA2}
\end{equation}
Now the only thing left to be specified is  $\tau_{k,k_{[\nu/2]:1}}^{(\nu)}$. The detailed derivation of $\tau_k^{\nu}$ could also be seen in
\cref{sec:derivation}. We have
\begin{equation}
\begin{aligned}
\tau_{0, k_{n:1}}^{(2n)}\left(t , T_{2n-1:1}, z_{0}\right)&=-\mathrm i V\left(\mathcal{E}_{k_n}, Q_{k_n,k_{n-1:1}}^{(2n-1)}\left(t, T_{2n-1:1}, z_{0}\right)\right),
\\
\tau_{k,k_{n:1}}^{(2n+1)}\left(t , T_{2n:1}, z_{0}\right)&=-\mathrm i{\left(\overline{V}\left(\mathcal{E}_{k}, Q_{0,k_{n:1}}^{(2n)}\left(t, T_{2n:1}, z_{0}\right)\right)\right)}.
\end{aligned}
\label{eq:hoppingcoefficient}
\end{equation}
We  refer to $\tau_{k,k_{[\nu/2]:1}}^{(\nu)}$ as the hopping coefficient associated with the $k$-th level,  and we will see in \cref{subsec:Stochasticinterpretation} that $\tau_{k,k_{[\nu/2]:1}}^{(\nu)}$ is used to define the hopping rate onto or away from the $k$-th level in the jump process that we will construct.

\subsection{Stochastic interpretation}
\label{subsec:Stochasticinterpretation}
In order to implement the Frozen Gaussian representation with an efficient  algorithm, we introduce its stochastic interpretation. 
We will construct a stochastic process $\widetilde{z}_t$, such that the above integral representation could be represented as taking expectations over all the path $\widetilde{z}_t$.

The stochastic process $\widetilde{z}_t=(z_t,l_t)$ is defined on the extended phase space $\mathbb{R}^{2m}\times\{0,1,\cdots,N\}$. Here $l_t\in\{0,1,\cdots,N\}$ represents the energy surface that the trajectory is currently on at time $t$. The evolution of $z_t$ is determined by  the energy surface specified by $l_t$:
\begin{equation}
\mathrm{d} z_{t}=
\left\{
\begin{aligned}
\left(p_{t},-\nabla_{q} U_0\left( q_{t}\right)\right) \mathrm{d} t,&\quad l_t=0, \\
\left(p_{t},-\nabla_{q} U_1\left( q_{t}\right)\right) \mathrm{d} t,&\quad l_t=1,\cdots,N.
\end{aligned}
\right.
\label{eq:Hanmiltonianflow}
\end{equation}
Here we take advantage of the fact that $\nabla_q \left(U_1(q)+\mathcal{E}_k\right)=\nabla_q U_1(q)$ since $\mathcal{E}_k$ is a constant. $z_t$ is coupled with $l_t$, which follows a jump process, with the following infinitesimal transition rate:
\begin{equation}
\mathbb{P}\left(l_{t+\delta t}=k' \mid l_{t}=k, z_{t}=z\right)=\delta_{kk'}+\lambda_{kk'}(z) \delta t+o(\delta t),
\label{eq:jumpprob}
\end{equation}
where 
\begin{equation}
 \lambda_{kk'}(z)=\left\{
\begin{aligned}
h\left|V\left(\mathcal{E}_k,q_t\right)\right|,\quad & k=0, k'=1,\cdots,N \\
\left|V\left(\mathcal{E}_k,q_t\right)\right|,\quad & k'=0, k=1,\cdots,N \\
-h
\sum_{k=1}^N\left|V\left(\mathcal{E}_k,q_t\right)\right|,\quad & k=k'=0\\
-\left|V\left(\mathcal{E}_k,q_t\right)\right|,\quad & k=k'=1,\cdots,N\\
0,\quad & {\text{ otherwise}}
\end{aligned}\right.\quad ,
\end{equation}
\begin{equation}
\begin{aligned}
 \lambda(z)&=
\left(\begin{array}{cccc}
\lambda_{00}(z) & \lambda_{01}(z)& \cdots & \lambda_{0N}(z)\\
\lambda_{10}(z) & \lambda_{11}(z) & \cdots & \lambda_{1N}(z) \\
\vdots &\vdots &\ddots &\vdots \\
\lambda_{N0}(z) & \lambda_{N1}(z) & \cdots & \lambda_{NN}(z) 
\end{array}\right)\\
&=
\left(\begin{array}{cccc}
-h
\sum_{k=1}^N\left|V\left(\mathcal{E}_k,q_t\right)\right| & h|V(\mathcal{E}_1,q_t)|& \cdots & h|V(\mathcal{E}_N,q_t)|\\
\left|V\left(\mathcal{E}_1,q_t\right)\right| & -\left|V\left(\mathcal{E}_1,q_t\right)\right| & \cdots & 0\\
\vdots &\vdots &\ddots &\vdots \\
\left|V\left(\mathcal{E}_N,q_t\right)\right| & 0 & \cdots & -\left|V\left(\mathcal{E}_N,q_t\right)\right|
\end{array}\right) .
\end{aligned}
\end{equation}
With $\lambda(z)$, we can write down now the forward Kolmogorov equation specified by the above hopping strategy
\begin{equation}
\frac{\partial}{\partial t} \rho_k(t,z)+\left\{h_{k}, \rho_k(t,z)\right\}=\sum_{m=0}^{N} \lambda_{m k}(z) \rho_k(t,z),
\label{eq:forwardkolmogorov}
\end{equation}
where $\rho_k(t,z)$ represents the phase space density function of the $k$-th level, and $\{\cdot\}$ is the Poisson bracket:
$$
\{h, F\}=\partial_{p} h \cdot \partial_{q} F-\partial_{q} h \cdot \partial_{p} F.
$$
Note that we have $\tau_{k,k_{[\nu/2]:1}}^{(\nu)}\left(t, T_{\nu: 1}, z_{0}\right)=\left|V(\mathcal{E}_k,Q_{k_{[\nu/2]:1}}^{(\nu-1)})\right|$. For $\nu=2n+1$, let us define $\widetilde{\lambda}_{k_{n:1}}^{(2 n+1)}\left(t, T_{2 n: 1}, z_{0}\right)$ for future purpose:
\begin{equation}
\widetilde\lambda_{k_{n:1}}^{(2 n+1)}\left(t, T_{2 n: 1}, z_{0}\right)=\sum_{k=1}^{N} h \left|\tau_{k,k_{n:1}}^{(2 n+1)}\left(t, T_{2 n: 1}, z_{0}\right)\right|.
\label{eq:lambda}
\end{equation}

Let $\mathbb{P}_{2n}(t,z_0)$ be the probability that there is $2n$ jumps in $[0,t]$. Let $\mathbb{P}_{2n+1}(t,k,z_0)$ be the probability that there is $(2n+1)$ jumps in $[0,t]$ and $l_t=k$. By definition, we have
\begin{equation}
1=\sum_{n=0}^{\infty} \mathbb{P}_{2 n}(t,z_0)+\sum_{n=0}^{\infty} \sum_{k=1}^{N} \mathbb{P}_{2 n+1}\left(t, k,z_0\right).
\end{equation}
By the properties of jumping process, we have
\begin{equation}
\begin{aligned}
\mathbb P_{2 n}(t,z_0)=\sum_{k_1,\cdots,k_n=1}^{N}\iint\cdots\int_{[0,t]^{2 n}}  &\mathrm dt_{2 n} \cdots  \mathrm dt_{1} \\ &\rho^{(2 n)}_{0,k_{n:1}}\left(T_{2 n:1},z_0\right)\Theta\left(0<t_{1}< \cdots<t_{2 n}< t\right),
\end{aligned}
\end{equation}
\begin{equation}
\begin{aligned}
   \mathbb P_{2 n+1}(t,k,z_0)=\sum_{k_1,\cdots,k_n,k=1}^{N}\iint\cdots\int_{[0,t]^{2 n+1}}  \mathrm dt_{2 n+1}&\cdots  \mathrm dt_{1} \\
\rho^{(2 n+1)}_{k,k_{n:1}}\left(T_{2 n+1:1},z_0\right)&\Theta\left(0<t_{1}< \cdots<t_{2 n+1}< t\right).
\end{aligned}
\end{equation}
where
\begin{equation}
\begin{aligned}
   &\rho^{(2 n)}_{0,k_{n:1}}\left(T_{2 n:1},z_0\right)  = \prod_{j=0}^{n}\left(\exp \left(-\int_{t_{2 j}}^{t_{2 j+1}} \mathrm d s_{2 j+1} \widetilde{\lambda}_{k_{n:1}}^{(2 n+1)}\left(s_{2 j+1}, T_{2 n: 1}, z_{0}\right)\right)\right)\\
   &\times  \left(\prod_{j=1}^{n} h\left|\tau_{k,k_{j-1:1}}^{(2 j-1)}\left(T_{2 j-1:1}\right)\right| \mathrm{e}^{\left(-\int_{t_{2 j-1}}^{t_{2 j}}\left|\tau_{0,k_{j:1}}^{(2 j)}\left(s_{2 j}, T_{2 j-1:1}\right)\right|\mathrm d s_{2 j}\right)}\left|\tau_{0,k_{j:1}}^{(2 j)}\left(T_{2 j:1}\right)\right|\right),
\end{aligned}
\end{equation}
and
\begin{equation}
\begin{aligned}
   &\rho^{(2 n+1)}_{k,k_{n:1}}\left(T_{2 n+1:1},z_0\right)  =
   h\left|\tau_{k,k_{j:1}}^{(2 n+1)}\left(T_{2 n+1:1}\right)\right|
   \exp \left(-\int_{t_{2n+1}}^{t} \mathrm d s \tau_{k,k_{n:1}}^{(2 n+2)}\left(s, T_{2 j-1: 1}, z_{0}\right)\right)
   \\
   &\times
   \prod_{j=0}^{n}\left(\exp \left(-\int_{t_{2 j}}^{t_{2 j+1}} \mathrm d s_{2 j+1} \widetilde{\lambda}_{k_{n:1}}^{(2 n+1)}\left(s_{2 j+1}, T_{2 n: 1}, z_{0}\right)\right)\right)\\
   &  \times\left(\prod_{j=1}^{n} h\left|\tau_{k,k_{j-1:1}}^{(2 j-1)}\left(T_{2 j-1:1}\right)\right| \mathrm{e}^{\left(-\int_{t_{2 j-1}}^{t_{2 j}}\left|\tau_{0,k_{j:1}}^{(2 j)}\left(s_{2 j}, T_{2 j-1:1}\right)\right|\mathrm d s_{2 j}\right)}\left|\tau_{0,k_{j:1}}^{(2 j)}\left(T_{2 j:1}\right)\right|\right).
\end{aligned}
\end{equation}
Based on these, we have :

\begin{equation}
\begin{aligned}
 u_{0,k_{n:1}}^{(2n)}(t, x)&=  \frac{1}{(2 \pi \varepsilon)^{3 m / 2}}  \int \mathrm{d} z_{0} \int_{0}^{t} \mathrm{~d} t_{2n} \int_{0}^{t_{2n}} \mathrm{~d} t_{2n-1} \cdots\int_{0}^{t_{2}} \mathrm{d} t_{1}\rho^{(2 n)}_{0,k_{n:1}}(T_{2n:1})\\
 &\times A_{0,k_{n:1}}^{(2n)}\left(t, T_{2n: 1}, z_{0}\right) \exp \left(\frac{\mathrm i}{\varepsilon} \Theta_{0,k_{n:1}}^{(2n)}\left(t, T_{2n: 1}, z_{0}, x\right)
\right) \Gamma_{0,k_{n:1}}^{(2n)}(t,T_{2n:1}),
\end{aligned}
\label{eq:ueven}
\end{equation}
\begin{equation}
\begin{aligned}
 &u_{k,k_{n:1}}^{(2n+1)}(t, x)= \frac{1}{(2 \pi \varepsilon)^{3 m / 2}}  \frac{1}{h}\int \mathrm{d} z_{0} \int_{0}^{t} \mathrm{~d} t_{2n+1} \int_{0}^{t_{2n}} \mathrm{~d} t_{2n} \cdots\int_{0}^{t_{2}} \mathrm{d} t_{1}\rho^{(2 n+1)}_{k,k_{n:1}}\left(T_{2 n+1:1}\right)\\
 & \times A_{k,k_{n:1}}^{(2n+1)}\left(t, T_{2n+1: 1}, z_{0}\right) \exp \left(\frac{\mathrm i}{\varepsilon} \Theta_{k,k_{n:1}}^{(2n+1)}\left(t, T_{2n+1: 1}, z_{0}, x\right)
\right)\Gamma_{k,k_{n:1}}^{(2n+1)}(t,T_{2n+1:1}),
\end{aligned}
\label{eq:uodd}
\end{equation}
where
\begin{equation}
\begin{aligned}
&\Gamma_{0,k_{n:1}}^{(2n)}(t,T_{2n:1},z_0)=\prod_{j=0}^{n}\left(\exp \left(\int_{t_{2 j}}^{t_{2 j+1}} \mathrm d s_{2 j+1} \widetilde{\lambda}_{k_{n:1}}^{(2 n+1)}\left(s_{2 j+1}, T_{2 n: 1}, z_{0}\right)\right)\right)
 \\
  &\times\prod_{j=1}^{n}\frac{\tau_{k_j,k_{j-1:1}}^{(2j-1)}\left(T_{2j-1: 1}\right)\tau_{{0,k_{j:1}}}^{(2j)}\left(T_{2j: 1}\right)\exp\left({\left(\int_{t_{2 j-1}}^{t_{2 j}}
  \left|\tau_{0,k_{j:1}}^{(2 j)}\left(s_{2 j}, T_{2 j-1:1}\right)\right|\mathrm d s_{2 j}\right)}\right)}{\left|\tau_{k_j,k_{j-1:1}}^{(2 j-1)}\left(T_{2 j-1:1}\right)\right|    
  \left|\tau_{0,k_{j:1}}^{(2 j)}\left(T_{2 j:1}\right)\right|},
\end{aligned}
\label{eq:Gammaeven}
\end{equation}
\begin{equation}
\begin{aligned}
 &\Gamma_{k,k_{n:1}}^{(2n+1)}(t,T_{2n+1:1},z_0)=\prod_{j=0}^{n}\left(\exp \left(\int_{t_{2 j}}^{t_{2 j+1}} \mathrm d s_{2 j+1} \widetilde{\lambda}_{k_{n:1}}^{(2 n+1)}\left(s_{2 j+1}, T_{2 n: 1}, z_{0}\right)\right)\right)\\
   &\times\exp \left(\int_{t_{2n+1}}^{t}  \tau_{0,k_{n+1:1}}^{(2 n+2)}\left(s, T_{2 j-1: 1}\right)\mathrm d s\right)\times \frac{\tau_{k,k_{n:1}}^{(2n+1)}\left(T_{2n+1: 1}\right)}{\left|\tau_{k,k_{n:1}}^{(2n+1)}\left(T_{2n+1: 1}\right)\right|}
 \\
  &\times\prod_{j=1}^{n}\frac{\tau_{k_j,k_{j-1:1}}^{(2j-1)}\left(T_{2j-1: 1}\right)\tau_{{0,k_{j:1}}}^{(2j)}\left(T_{2j: 1}\right)\exp\left({\left(\int_{t_{2 j-1}}^{t_{2 j}}
  \left|\tau_{0,k_{j:1}}^{(2 j)}\left(s_{2 j}, T_{2 j-1:1}\right)\right|\mathrm d s_{2 j}\right)}\right)}{\left|\tau_{k,k_{j-1:1}}^{(2 j-1)}\left(T_{2 j-1:1}\right)\right|    
  \left|\tau_{0,k_{j:1}}^{(2 j)}\left(T_{2 j:1}\right)\right|}.
\end{aligned}
\label{eq:Gammaodd}
\end{equation}

Finally, we are ready to present the Frozen Gaussian representation as an average over stochastic path.
Let
\begin{equation}
    \mathbf U(t,x)=\left(
    \begin{array}{c}
     u_0(t,x) \\ u_1(t,x) \\ \cdots \\u_N(t,x)
    \end{array}
    \right),
\end{equation}
And let $\mathbf U_{\text{FG}}(t,x)$ 
be the Frozen Gaussian representation of $\mathbf U(t,x)$.  Rewriting \cref{eq:ueven}, \cref{eq:uodd} with \cref{eq:Gammaeven} and \cref{eq:Gammaodd}, we have 
\begin{equation}
   \mathbf{U}_{\text{FG}}(t,x)=\mathbb{E}_{\widetilde z_t}\left(
   \Lambda(t,x;\widetilde z_t)
    \right),
    \label{eq:expectation}
\end{equation}
where 
\begin{equation}
   \Lambda(t,x;\widetilde z_t)= 
    W(t,x;\widetilde z_t){\mathbb{V}}^{(k)},
    \label{eq:wavepacket}
\end{equation}
\begin{equation}
   W(t,x;\widetilde z_t)= \frac{1}{(2\pi\varepsilon)^{3m/2}}\frac{1}{\pi(z_0)}A_{k,k_{[{\nu}/{2}]:1}}^{(\nu)}\exp\left(\frac{\mathrm i}{\varepsilon}\Theta_{k,k_{[{\nu}/{2}]:1}}^{(\nu)}\right)
    \Gamma_{k,k_{[\nu/2]:1}}^{(\nu)}.
    \label{eq:wavepacketscalar}
\end{equation}
Here $k$ is the index of surface that the wave packet is currently on, $\pi(z_0)$ is the probability density that $z_0$ is sampled from, and $\mathbb{V}^{(\nu)}$  
is an  $(N+1)\times 1$ vector, its $j$-th component is ($j=0,1,
\cdots,N$)
\begin{equation}
    {\mathbb{V}}_j^{(\nu)}=\left\{
    \begin{aligned}
    \delta_{0j},\quad \nu \text{ is even}, \\
    \frac{1}{h}\delta_{k_{\lceil{\nu}/{2}\rceil}j},\quad \nu \text{ is odd}.
    \end{aligned}
    \right.
\end{equation}
Note that $\Lambda(t,x;\widetilde z_t)$ is a $(N+1)$-dimensional vector-valued function  and $W(t,x;\widetilde z_t)$ is a scalar-valued function. 

We emphasize that  $\Lambda(t,x;\widetilde z_t)$ could be viewed as a functional of the stochastic trajectory $\widetilde z_t$. As a result, according to \cref{eq:expectation}, the wavefunction $\mathbf{U}_{\text{FG}}$ is in fact the average of $\Lambda(t,x;\widetilde z_t)$ \cref{eq:wavepacket} over the stochastic path $\widetilde z_t$. From an algorithmic perspective, the sampling of the stochastic path $\widetilde z_t $ is achieved by first sampling its initial value $z_0\sim \pi(z)$ and then a stochastic evolution subject to the Hamiltonian flow \cref{eq:Hanmiltonianflow} and the jump process \cref{eq:jumpprob}.

\begin{remark}
Later We'll use $\Lambda_k(t,x;\widetilde z_t)$ to represent the $k$-th component of $\Lambda$ for $k=0,1,\cdots,N$.  
We will also use $\mathbf U_{k,\operatorname{FG}}(t,x)$ 
to represent the $k$-th component of $\mathbf U_{\operatorname{FG}}(t,x)$ 
for $k=0,1,\cdots,N$.
\end{remark}

\subsection{FGS Algorithm}
\label{subsec:algorithm}
Based on the stochastic interpretation \cref{eq:expectation}, we are able to constuct the  Frozen Gaussian Sampling (FGS) algorithm, as illustrated in \cref{fig:FGS}.

\begin{figure}[h]
    \centering
    \includegraphics[width = 120mm]{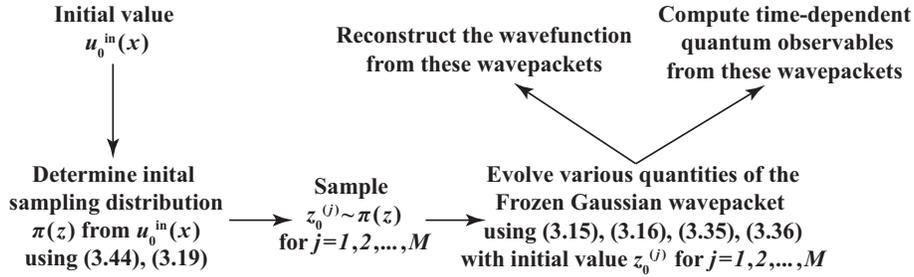}
    \caption{Schematics of Frozen Gaussian Sampling algorithm for metal surfaces}
    \label{fig:FGS}
\end{figure}

After drawing $M$ independent initial wave packets by a properly chosen phase space distribution, each wavepacket evolves subject to the Hamiltonian flow \cref{eq:Hanmiltonianflow} and hops between energy surfaces following the jump process \cref{eq:jumpprob}.
Then at time $t$, we can reconstruct the solution for system \cref{eq:discretize_u} with the average of $M$ wave packets:
\begin{equation}
   \mathbf{U}_{\text{FG}}(t,x)\approx \frac{1}{M}\sum_{j=1}^M\Lambda^{(j)}\left(t,x;\widetilde{z}_t^{(j)}\right).
\end{equation}
Here the superscript $(j)$ indicates the $j$-th wavepacket that we have sampled.
We summarize these procedures in the following \cref{al:mainalgorithm}.
\begin{algorithm}
\caption{Frozen Gaussian Sampling algorithm for nonadiabatic quantum dynamics at metal surfaces}
\label{al:mainalgorithm}
\begin{algorithmic}[1]
\STATE Initial Sampling: generate $M$ i.i.d samples $\left\{z_0^{(j)}\right\}_{j=1}^M$ according to the distribution $\pi$, and then calculate the corresponding $A_0^{(0)}\left(0,z_0^{(j)}\right)$ using equation \cref{eq:dynamicsforamplitude2}.
\STATE Trajectory evolving: for each $j=1,\cdots,M$, with $z_0^{(j)}$ and $A_0^{(0)}\left(0,z_0^{(j)}\right)$,  calculate $A_{k^j_{[{\nu_j}/{2}]:1}}^{(\nu_j)}$, $P_{k^j_{[{\nu_j}/{2}]:1}}^{(\nu_j)}$,
$Q_{k^j_{[{\nu_j}/{2}]:1}}^{(\nu_j)}$ and $S_{k^j_{[{\nu_j}/{2}]:1}}^{(\nu_j)}$ using \cref{eq:groupdynamics}, and then calculate $\Theta_{k^j_{[{\nu_j}/{2}]:1}}^{(\nu_j)}$ using \cref{eq:Thetarelation} and $\Gamma_{k^j,k^j_{[\nu_j/2]:1}}^{(\nu_j)}$ using \cref{eq:Gammaeven,eq:Gammaodd}.
\STATE Wavefunction constructing: calculate $\frac{1}{M}\sum_{j=1}^M\Lambda^{(j)}\left(t,x;\widetilde{z}_t^{(j)}\right)$, where the wave packet $\Lambda^{(j)}$ 
 is constructed using \cref{eq:wavepacket}.
\end{algorithmic}
\end{algorithm}

We can also calculate time-dependent physical observables. We would like to emphasize that this could be done without the reconstruction of wavefunctions. Consider a physical observable $\hat O$, which corresponds to a compact operator on $L^2(\mathbb{R}^m,\mathbb{C}^{N+1})$, we have
\begin{equation}
    \hat{O}(t) = \frac{1}{M^2}\sum_{j,j'=1}^M \left\langle\Lambda^{(j)}\left(t,x;\widetilde{z}_t^{(j)}\right)\right|\hat O\left|\Lambda^{(j')}\left(t,x;\widetilde{z}_t^{(j')}\right)\right\rangle.
\end{equation}
It should be noted that since $\Lambda^{(j)}\left(t,x,z_0^{(j)}\right)$ has the frozen Gaussian form, therefore for a wide class of observables, the integral in the double sum can be evaluated analytically or approximately by Gaussian integral formulas.

We often consider the Gaussian type initial conditions, where the initial value $u_{0,\text{in}}$ is a Gaussian wave packet:
\begin{equation}
u_{0,\mathrm{in}}(x)=\left(\prod_{j=1}^{m} a_{j}\right)^{\frac{1}{4}}(\pi \varepsilon)^{-\frac{m}{4}} \exp \left(\frac{\mathrm{i}}{\varepsilon} \tilde{p} \cdot(x-\tilde{q})\right) \exp \left(-\sum_{j=1}^{m} \frac{a_{j}}{2 \varepsilon}\left(x_{j}-\tilde{q}_{j}\right)^{2}\right).
\label{eq:initialgaussian}
\end{equation}
Here $\tilde{p}, \tilde{q} \in \mathbb{R}^{m}$ and $a_{j}(1 \leq j \leq m)$ are positive. In this case, we have
\begin{equation}
\begin{aligned}
A_0^{(0)}(0, q, p)=2^{m}(\pi \varepsilon)^{\frac{m}{4}} \prod_{j=1}^{m} & {\left[\left(\frac{\sqrt{a_{j}}}{1+a_{j}}\right)^{\frac{1}{2}} \exp \left(-\frac{\left(\tilde{p}_{j}-p_{j}\right)^{2}+a_{j}\left(\tilde{q}_{j}-q_{j}\right)^{2}}{2\left(1+a_{j}\right) \varepsilon}\right)\right.} \\
&\left.\times \exp \left(\frac{\mathrm{i}\left(a_{j} \tilde{q}_{j}+q_{j}\right)\left(\tilde{p}_{j}-p_{j}\right)}{\left(1+a_{j}\right) \varepsilon}+\frac{\mathrm{i}\left(p_{j} q_{j}-\tilde{p}_{j} \tilde{q}_{j}\right)}{\varepsilon}\right)\right],
\end{aligned}
\end{equation}
A suitable initial sampling distribution for Gaussian wave packet is \cite{lu2018frozen,xie2021frozen}
\begin{equation}
    \pi(q,p)=\frac{A_0^{(0)}(0,q,p)}{\int_{\mathbb{R}^{2m}}\mathrm{d}q\mathrm{d}p|A_0^{(0)}(0,q,p)|}.
\end{equation}
therefore we have \begin{equation}
\pi(q, p)=(2 \pi \varepsilon)^{-m} \prod_{j=1}^{m}\left(\frac{\sqrt{a_{j}}}{1+a_{j}}\right) \exp \left(-\sum_{j=1}^{m} \frac{\left(\tilde{p}_{j}-p_{j}\right)^{2}+a_{j}\left(\tilde{q}_{j}-q_{j}\right)^{2}}{2\left(1+a_{j}\right) \varepsilon}\right).
\end{equation}
It is clear that $\pi$ is a Gaussian distribution and is extremely easy to sample. In other words, we have the initial condition for $\cref{eq:forwardkolmogorov}$:
\begin{equation}
    \rho_0(0,q,p)=\pi(q,p), \quad \rho_k(0,q,p)=0,\quad k=1,\cdots,N.
\end{equation}
For other intial condition, we refer to \cite{xie2021frozen} for a broad discussion on proper choices of the initial sampling measure.

\section{Analysis}
\label{sec:analysis}
In this section, we present the numerical analysis of FGS. We focus on the case of Gaussian initial data \cref{eq:initialgaussian}. 
\subsection{Main result}
After sampling $M$ wave packets to approximate $U_{\operatorname{FG}}$, the sampling error should be evaluated using the following weighted $L^2$ error:

\begin{equation}
  \mathbf{e}_N \left(\left\{\widetilde{z}_t^{(j)}\right\}_{j=1}^M\right) = \left\|R_N\left(\frac{1}{M}\sum_{j=1}^M\Lambda^{(j)}\left(t,x;\widetilde{z}_t^{(j)}\right)-{\mathbf{U}}_{\text{FG}}(t,x)\right)\right\|_{L^2\left(\mathbb{R}^m,\mathbb{C}^{N+1}\right)},
  \label{eq:weightederror}
\end{equation}
where $R_N$ is the following $(N+1)\times (N+1)$ weight matrix:
\begin{equation}
  R_N=\operatorname{diag}(1,h,h,\cdots,h).
\end{equation}
We adopt the weighted $L^2$ error \cref{eq:weightederror} due to the fact that the $(N+1)$-level Schr\"{o}dinger equation system is obtained by discretizing the continuum spectrum. In other words, the $h$ in the weighted matrix should be understood as a numerical integration along the energy spectrum $[\mathcal{E}_a,\mathcal{E}_b]$.

Since FGS is a stochastic method, its numerical error should be evaluated as the following mean square error
\begin{equation}
\begin{aligned}
\mathrm{Err}_M(t,N)=& \left[\mathbb{E}_{\left\{\widetilde{z}_t^{(j)}\right\}_{j=1}^M}\left(\left(\mathbf{e}_N \left(\left\{\widetilde{z}_t^{(j)}\right\}_{j=1}^M\right)\right)^2\right)\right]^{\frac{1}{2}},
\end{aligned}
\end{equation}
where $L^2\left(\mathbb{R}^m,\mathbb{C}^{N+1}\right)$ is the $(N+1)$-dimensional-valued $L^2$-space on $\mathbb{R}^m$.
In fact, since all $z_0^{(j)}$ are sampled independently, we have
\begin{equation}
\begin{aligned}
{\mathrm{Err}}_M(t,N)  
 & = \frac{1}{\sqrt{M}} {\mathrm{Err}}_1(t,N).
 \label{eq:halforder}
\end{aligned}
\end{equation}
 To proceed with the error estimate, we make the following assumption on the potential functions: 
\begin{assumption}
 Each energy surface $U_k(q) \in C^{\infty}\left(\mathbb{R}^{m}\right), k \in\{0,1\}$ and satisfies the following subquadratic condition:
 \begin{equation}
\sup _{q \in \mathbb{R}^{m}}\left|\partial_{\alpha} U_{k}(q)\right| \leq C_{E}, \quad \forall|\alpha|=2.
\end{equation}
If $U_k(q)$ depends on $\varepsilon$, we require the inequality to hold uniformly for all $\varepsilon$.
\label{asption}
\end{assumption}
Our main numerical analysis result is that:
\begin{theorem}
For system \cref{eq:discretize_u} 
with  initial value \cref{eq:initialgaussian},  let $M$ be the sample size. With  \cref{asption}, the sampling error ${\mathrm{Err}}_M(t,N)$ of the FGS algorithm (\cref{al:mainalgorithm}) satisfies the following inequality:
\begin{equation}
{\mathrm{Err}}_M(t,N)\leq \frac{1}{\sqrt{M}}
K(t) \prod_{j=1}^{m}\left(\frac{1+a_{j}}{\sqrt{a_{j}}}\right)^{\frac{1}{2}},
\end{equation}
where $K(t)$ is independent of both the semiclassical parameter $\varepsilon$ and the number of metal orbitals $N$ .
\label{thm:main}
\end{theorem}
This theorem indicates that the sample size required by the FGS algorithm does not increase  when $N$ gets larger.
In fact, recall \cref{eq:wavepacketscalar,eq:wavepacket}, we observe that
\begin{equation}
\begin{aligned}
 \mathrm{Err}_1(t,N)&= \left(\mathbb{E}_{\widetilde{z}_t}\left(\left\|R_N\left(\Lambda(t,x;{\widetilde{z}_t})-{\mathbf{U}}_{\text{FG}}(t,x)\right)\right\|^2_{L^2\left(\mathbb{R}^m,\mathbb{C}^{N+1}\right)}\right)\right)^{\frac{1}{2}} \\
 &\leq  \left(\mathbb{E}_{\widetilde{z}_t}\left(\left\|R_N\left(\Lambda(t,x;{\widetilde{z}_t})\right)\right\|^2_{L^2\left(\mathbb{R}^m,\mathbb{C}^{N+1}\right)}\right)\right)^{\frac{1}{2}} \\
 & =\left(\mathbb{E}_{\widetilde{z}_t}\left(\left\|W(t,x;{\widetilde{z}_t})\right\|^2_{L^2\left(\mathbb{R}^m\right)}\right)\right)^{\frac{1}{2}}.
\end{aligned}
\label{eq:err1}
\end{equation}
Here this inequality uses the fact that the variance of a stochastic variable is not bigger than its secondary moment.
The value of 
$\left\|W(t,x;{\widetilde{z}_t})\right\|^2_{L^2\left(\mathbb{R}^m\right)}$ 
is a function of the stochastic trajectory $\widetilde{z}_t$. Let us introduce the following notation:
\begin{equation}
    f({\widetilde{z}_t})= \left\|W(t,x,z_0 )\right\|^2_{L^2\left(\mathbb{R}^m\right)}.
\end{equation}
Then along with \cref{eq:halforder}, it suffices to only prove the following result:
\begin{theorem}
With the same assumptions in \cref{thm:main}, we have the following inequality:
\begin{equation}
{\mathbb{E}_{\widetilde{z}_t}\left(f(\widetilde{z}_t)\right)}\leq 
(K(t))^2 \prod_{j=1}^{m}\left(\frac{1+a_{j}}{\sqrt{a_{j}}}\right),
\end{equation}
where $K(t)$ is independent of both the semiclassical parameter $\varepsilon$ and the number of metal orbitals $N$.
\label{thm:theotherthm}
\end{theorem}

\subsection{Proof of \cref{thm:theotherthm}}
We first introduce a lemma that gives an upper bound for $\left| \Gamma_{k,k_{[\nu/2]:1}}^{(\nu)}(t,T_{\nu:1})\right|$:
\begin{lemma}[Upper bound for $\Gamma_{k,k_{[\nu/2]:1}}^{(\nu)}(t,T_{\nu:1},z_0)$]
Let $C=\max(1,\mathcal{E}_b-\mathcal{E}_a)$, $M_0=\max_{\mathcal{E}\in [\mathcal{E}_a,\mathcal{E}_b]}\left|V(\mathcal{E})\right|$.
we have
\begin{equation}
    \left|\Gamma_{k,k_{[\nu/2]:1}}^{(\nu)}(t,T_{\nu:1},z_0)\right|\leq \exp({CM_0t}),
\end{equation}
for any $\nu\geq 0$, any $k,k_{[\nu/2]:1}\in\{0,1,\cdots,N\}$, any $z_0\in\mathbb{R}^{2m}$ and any $0<t_1<\cdots<t_{\nu}<t$.
\label{lemma:gamma}
\end{lemma}
\begin{proof}
From \cref{eq:hoppingcoefficient}, we can see that   $\left|\tau_{k,k_{[\nu/2]:1}}^{(\nu)}\right| \leq M_0$
and according to \cref{eq:lambda}, $\left|\widetilde\lambda_{k_{n:1}}^{(2 n+1)}\left(t, T_{2 n: 1}, z_{0}\right)\right|\leq (\mathcal E_b-\mathcal E_a)M_0$. Therefore, $\left|\tau_{k,k_{[\nu/2]:1}}^{(\nu)}\right| \leq CM_0$, $\left|\widetilde\lambda_{k_{n:1}}^{(2 n+1)}\right|\leq CM_0$. For $\nu$ being odd, according to \cref{eq:Gammaodd}, we have
\begin{equation}
  \begin{aligned}
 &\left|\Gamma_{k,k_{n:1}}^{(2n+1)}(t,T_{2n+1:1},z_0)\right|=\prod_{j=0}^{n}\left(\exp \left(\int_{t_{2 j}}^{t_{2 j+1}}  \widetilde{\lambda}_{k_{n:1}}^{(2 n+1)}\left(s_{2 j+1}, T_{2 n: 1}, z_{0}\right)\mathrm d s_{2 j+1}\right)\right)\\
   &\times\exp \left(\int_{t_{2n+1}}^{t}  \tau_{0,k_{n+1:1}}^{(2 n+2)}\mathrm d s\right)
  \times\prod_{j=1}^{n}\exp\left({\left(\int_{t_{2 j-1}}^{t_{2 j}}
  \left|\tau_{0,k_{j:1}}^{(2 j)}\left(s_{2 j}, T_{2 j-1:1}\right)\right|\mathrm d s_{2 j}\right)}\right)\\
  &\leq \prod_{j=0}^{n}\exp \left(\int_{t_{2 j}}^{t_{2 j+1}} CM_0\mathrm d s_{2 j+1} \right)\exp \left(\int_{t_{2n+1}}^{t}CM_0 \mathrm d s \right)
  \prod_{j=1}^{n}\exp{\left(\int_{t_{2 j-1}}^{t_{2 j}}
  CM_0\mathrm d s_{2 j}\right)} \\
  & \leq \exp(CM_0t).
\end{aligned}
\label{eq:gammabound}
\end{equation}
For $\nu$ being even, based on \cref{eq:Gammaeven}, we can prove the same result in a similar way as in \cref{eq:gammabound}.
\end{proof}

Using \cref{lemma:gamma}, we can give an upper bound of $\mathbb{E}_{\widetilde{z}_t}\left(f(\widetilde{z}_t)\right)$:
\begin{lemma}
With the same assumptions in \cref{thm:main}, we have the following inequality:
\begin{equation}
\mathbb{E}_{\widetilde{z}_t}\left(f(\widetilde{z}_t)\right)\leq 
\frac{\exp(2CM_0t)}{(2 \pi \varepsilon)^{\frac{5}{2} m}} \int \mathrm d z_{0} \frac{1}{\pi({z_0})} \mathbb{E}_{\widetilde{z}_{t}| z_{0}}
\left(\left|A_{k,k_{[{\nu}/{2}]:1}}^{(\nu)} \left(t ,T_{\nu:1}, z_{0}\right)\right|^{2}\right),
\end{equation}
where $C,M_0$ is defined the same way as in \cref{lemma:gamma}.
\label{lemma:temporary}
\end{lemma}
\begin{proof}
With \cref{eq:wavepacket} and \cref{lemma:gamma}, we have
\begin{equation}
    \begin{aligned}
     \left|W(t,x,z_0)\right|^2&=\frac{1}{(2\pi\varepsilon)^{3m}}\frac{\left|A_{k,k_{[{\nu}/{2}]:1}}^{(\nu)} \right|^{2}}{\pi(z_0)^2}
     \exp \left(-\frac{\left|x-Q_{k_{[{\nu}/{2}]:1}}^{(\nu)}\right|^{2}}{\varepsilon}\right)\left|\Gamma_{k,k_{[\nu/2]:1}}^{(\nu)}\right|^2 \\
     &\leq \frac{1}{(2\pi\varepsilon)^{3m}}\frac{\left|A_{k,k_{[{\nu}/{2}]:1}}^{(\nu)} \right|^{2}}{\pi(z_0)^2}
     \exp \left(-\frac{\left|x-Q_{k_{[{\nu}/{2}]:1}}^{(\nu)}\right|^{2}}{\varepsilon}\right)\exp(2CM_0t).
    \end{aligned}
\end{equation}
Since $\int_{\mathbb{R}^m}\exp (-{|x-Q_{k_{[{\nu}/{2}]:1}}^{(\nu)}|^{2}}/{\varepsilon})\mathrm d x=(2\pi\varepsilon)^{m/2}$, we have 
\begin{equation}
    \begin{aligned}
    \mathbb{E}_{\widetilde{z}_t}\left(f(\widetilde{z}_t)\right)
   & =
   \mathbb{E}_{\widetilde{z}_t}\left( \int_{\mathbb{R}^m} \left|W(t,x,z_0)\right|^2\mathrm d x\right) \\
   & \leq\frac{\exp(2CM_0t)}{(2\pi\varepsilon)^{5m/2}}  \mathbb{E}_{\widetilde{z}_t}\left(
   \frac{\left|A_{k,k_{[{\nu}/{2}]:1}}^{(\nu)} \right|^{2}}{\pi(z_0)^2}\right)
   \\
   & =\frac{\exp(2CM_0t)}{(2\pi\varepsilon)^{5m/2}} \mathbb{E}_{z_0\sim\pi}\left(\frac{1}{\pi(z_0)^2}\mathbb{E}_{\widetilde{z}_t|z_0}\left(
  {\left|A_{k,k_{[{\nu}/{2}]:1}}^{(\nu)} \right|^{2}}\right)\right) \\
  & = \frac{\exp(2CM_0t)}{(2\pi\varepsilon)^{5m/2}}
  \int\mathrm{d}z_0 \frac{1}{\pi(z_0)}\mathbb{E}_{\widetilde{z}_t|z_0}\left(
  {\left|A_{k,k_{[{\nu}/{2}]:1}}^{(\nu)} \right|^{2}}\right).
    \end{aligned}
\end{equation}
\end{proof}

Now, with the following lemma, we can finally prove \cref{thm:theotherthm}.
\begin{lemma}
Given time $t$, with \cref{asption}, for any $\nu\geq 0$, any $k,k_{[\nu/2]:1}\in\{0,1,\cdots,N\}$, any $z_0\in\mathbb{R}^{2m}$ and  any $0<t_1<\cdots<t_{\nu}<t$, there exist $G(t)$ such that
\begin{equation}
    \frac{\left|A_{k,k_{[{\nu}/{2}]:1}}^{(\nu)}\left(t ,T_{n:1}, z_{0}\right) \right|}{\left|A_0^{(0)}(0,z_0)\right|}\leq G(t),
\end{equation}
where $G(t)$ is independent of $N$ and $\varepsilon$.
\label{lemma:growth}
\end{lemma}
The proof of \cref{lemma:growth} is described later. Let us first demonstrate how to use \cref{lemma:growth} to prove \cref{thm:theotherthm}.
\begin{proof}[Proof of \cref{thm:theotherthm} using \cref{lemma:growth}]

With \cref{lemma:temporary} and \cref{lemma:growth}, we have:
\begin{equation}
    \begin{aligned}
    \mathbb{E}_{\widetilde{z}_t}\left(f(\widetilde{z}_t)\right)
  & \leq  \frac{\exp(2CM_0t)}{(2\pi\varepsilon)^{5m/2}}
  \int\mathrm{d}z_0 \frac{1}{\pi(z_0)}\mathbb{E}_{\widetilde{z}_t|z_0}\left(
  {\left|A_{k,k_{[{\nu}/{2}]:1}}^{(\nu)} \right|^{2}}\right) \\
  &\leq \frac{\exp(2CM_0t)G(t)^2}{(2\pi\varepsilon)^{5m/2}}
  \int\mathrm{d}z_0 \frac{\left|A_0^{(0)}(0,z_0)\right|^2}{\pi(z_0)}.
    \end{aligned}
\end{equation}
Recall that $\pi\left(z_{0}\right)=\frac{\left|A_0^{(0)}\left(0, z_{0}\right)\right|}{\int\left|A_0^{(0)}\left(0, z_{0}\right)\right| \mathrm d z_{0}}$, let $K(t)=\exp(CM_0t)G(t)$, then  
\begin{equation}
    \begin{aligned}
    \mathbb{E}_{\widetilde{z}_t}\left(f(\widetilde{z}_t)\right)
  &\leq \frac{K(t)^2}{(2\pi\varepsilon)^{5m/2}}
  \int\mathrm{d}z_0 \frac{\left|A_0^{(0)}(0,z_0)\right|^2}{\pi(z_0)}= \frac{\left(\int\left|A_0^{(0)}\left(0, z_{0}\right)\right| \mathrm d z_{0}\right)^{2}}{(2 \pi \varepsilon)^{5 m / 2}}.
    \end{aligned}
\end{equation}
Since
$$
\int_{\mathbb{R}^{2 m}}\left|A_0^{(0)}\left(0, z_{0}\right)\right| \mathrm{d} z_{0}=2^{2 m}(\pi \varepsilon)^{\frac{5 m}{4}} \prod_{j=1}^{m}\left(\frac{1+a_{j}}{\sqrt{a_{j}}}\right)^{\frac{1}{2}},
$$
then we have 
\begin{equation}
\mathbb{E}_{\widetilde{z}_t}\left(f(\widetilde{z}_t)\right)
  \leq
2^{\frac{3 m}{2}}\left(K(t)\right)^{2} \prod_{j=1}^{m}\left(\frac{1+a_{j}}{\sqrt{a_{j}}}\right).
\end{equation}
\end{proof}
Note that using \cref{eq:halforder} and \cref{eq:err1}, \cref{thm:theotherthm} immediately yields \cref{thm:main}.
\subsection{Proof of \cref{lemma:growth}}
Now we only need to provide a proof for \cref{lemma:growth}. 

Note that the dynamics of $Q_{k,k_{[{\nu}/{2}]:1}}^{(\nu)}\left(t ,T_{\nu:1}, z_{0}\right)$
and
$P_{k,k_{[{\nu}/{2}]:1}}^{(\nu)}\left(t ,T_{\nu:1}, z_{0}\right)$
doesn't depend on the value of $k$ and $k_{[{\nu}/{2}]:1}$. This can be seen in \cref{eq:Hanmiltonianflow} where we take advantage of the fact that $U_k(x)$ only differs from each other by a constant for $k=1,\cdots,N$. In other words, the Hamiltonian flow only depends on the hopping time $T_{\nu: 1}=\left\{t_{\nu}, \ldots, t_{1}\right\}$. We denote the map on the phase space from initial time 0 to time $s$ by $\kappa_{s, T_{\nu: 1}}$ ($s\in[0,t]$):
\begin{equation}
\begin{aligned}
\kappa_{s, T_{\nu: 1}}: & \mathbb{R}^{2 m} \rightarrow \mathbb{R}^{2 m} \\
&(q, p) \longmapsto\left(Q^{\kappa_{s, T_{\nu: 1}}}(q, p), P^{\kappa_{s, T_{\nu: 1}}}(q, p)\right),
\end{aligned}
\end{equation}
such that
\begin{equation}
\begin{array}{l}
\left(Q^{\kappa_{s, T_{j: 1}}}(q, p), P^{\kappa_{s, T_{j: 1}}(q, p)}\right) \\
\quad=\left\{\begin{array}{ll}
\left(Q^{(0)}(t, q, p), P^{(0)}(t, q, p)\right), & t \leq t_{1} \\
\left(Q^{(j)}\left(t, T_{j: 1}, q, p\right), P^{(j)}\left(t, T_{i: 1}, q, p\right)\right), & t \in\left[t_{j}, t_{j+1}\right], j \in\{1, \ldots, \nu\} \\
\left(Q^{(\nu)}\left(t, T_{\nu: 1}, q, p\right), P^{(\nu)}\left(t, T_{\nu: 1}, q, p\right)\right), & t \geq t_{\nu}
\end{array}\right.  .
\end{array}
\end{equation}
Here we omit the subscript $k,k_{[{\nu}/{2}]:1}$ for $Q^{(\nu)}$ and $P^{(\nu)}$.

For a transformation $\kappa$ on the phase space $\mathbb{R}^{2m}$, we can define its Jacobian matrix as
\begin{equation}
J^{\kappa}(q, p)=\left(\begin{array}{ll}
\left(\partial_{q} Q^{\kappa}\right)^{T}(q, p) & \left(\partial_{p} Q^{\kappa}\right)^{T}(q, p) \\
\left(\partial_{q} P^{\kappa}\right)^{T}(q, p) & \left(\partial_{p} P^{\kappa}\right)^{T}(q, p)
\end{array}\right) .
\end{equation}
We also define
\begin{equation}
Z^{\kappa}(q, p)=\partial_{z}\left(Q^{\kappa}(q, p)+i P^{\kappa}(q, p)\right),\quad \partial_{z}=\partial_{q}-i \partial_{p} .
\end{equation}
Let us quote the following lemma from \cite{lu2018frozen,lu2010frozen} that states the property of the Hamiltonian flow:
\begin{lemma}
Given $t>0$, with \cref{asption}, for any $\nu$ and any $T_{\nu:1}$, the asscociated map $\kappa_{t, T_{\nu: 1}}$ and $J^{\kappa_{t, T_{\nu: 1}}}(q, p)$, $Z^{\kappa_{t, T_{\nu: 1}}}(q, p)$ has the follow properties:
\begin{enumerate}
    \item $\kappa_{t, T_{\nu: 1}}$ is a canonical transformation, i.e.
    \begin{equation}
\left(J^{\kappa_{t, T_{\nu: 1}}}\right)^{T}\left(\begin{array}{cc}
0 & I_{m} \\
-I_{m} & 0
\end{array}\right) J^{\kappa_{t, T_{\nu: 1}}}=\left(\begin{array}{cc}
0 & I_{m} \\
-I_{m} & 0
\end{array}\right),
\end{equation}
Here $I_{m}$ is the $m \times m$ identity matrix.
\item For any $k \in \mathbb{N}$, there exists a constant $C_{k}$ such that
\begin{equation}
\sup _{(q, p) \in \mathbb{R}^{2m}} \max _{\left|\alpha_{p}\right|+\left|\alpha_{q}\right| \leq k}\left|\partial_{q}^{\alpha_{q}} \partial_{p}^{\alpha_{p}}\left[J^{\kappa_{t, T_{\nu: 1}}}(q, p)\right]\right| \leq C_{k}(t),
\end{equation}
$C_{k}(t)$ is dependent on $U_0$ and $U_1$.
\item $Z^{\kappa_{t, T_{\nu: 1}}}$ is an invertible $m\times m$ matrix, and for any $k \in \mathbb{N}$, there exists a constant $C_{k}$ such that
\begin{equation}
\sup _{(q, p) \in K} \max _{\left|\alpha_{p}\right|+\left|\alpha_{q}\right| \leq k}\left|\partial_{q}^{\alpha_{q}} \partial_{p}^{\alpha_{p}}\left[\left(Z^{\kappa_{t, T_{j: 1}}}(q, p)\right)^{-1}\right]\right| \leq C_{k}(t),
\end{equation}
$C_{k}(t)$ is dependent on $U_0$ and $U_1$.
\end{enumerate}
\label{lemma:Hamiltonianflow}
\end{lemma}
Since $\kappa$ doesn't depend on the value of $k_{[{\nu}/{2}]:1}$, therefore the value of $C_k(t)$ in \cref{lemma:Hamiltonianflow} doesn't depend on $N$ at all.

Based on this lemma we are ready to give a proof for \cref{lemma:growth}. 
\begin{proof}[Proof of \cref{lemma:growth}]
Recall that on $[t_{\nu},t_{\nu+1}]$, we have
$$
\frac{\mathrm{d}}{\mathrm{d} t} A_{k,k_{[{\nu}/{2}]:1}}^{(\nu)}= \frac{1}{2} A_{k,k_{[{\nu}/{2}]:1}}^{(\nu)} \operatorname{tr}\left(\left(Z^{(\nu)}\right)^{-1}\left(\partial_{z} P_{k_{[{\nu}/{2}]:1}}^{(\nu)}-i \partial_{z} Q_{k_{[{\nu}/{2}]:1}}^{(\nu)} \nabla_{Q}^{2} \widetilde U\right)\right), 
$$
where
  $$ \widetilde{U}(t,Q)=\left\{
    \begin{aligned}
    U_0(Q),\quad &t\in [t_{2j}, t_{2j+1}]
    \\
    U_1(Q)+\mathcal{E}_{k_j},\quad &t\in [t_{2j-1}, t_{2j}]
    \end{aligned}
    \right. .$$
Therefore on $[t_{\nu},t_{\nu+1}]$, we have
\begin{equation}
    \frac{\mathrm{d}}{\mathrm{d} t}\left| A_{k,k_{[{\nu}/{2}]:1}}^{(\nu)}\right|
    \leq  \frac{1}{2} \left|A_{k,k_{[{\nu}/{2}]:1}}^{(\nu)}\right| \left|\operatorname{tr}\left(\left(Z^{(\nu)}\right)^{-1}\left(\partial_{z} P_{k_{[{\nu}/{2}]:1}}^{(\nu)}-i \partial_{z} Q_{k_{[{\nu}/{2}]:1}}^{(\nu)} \nabla_{Q}^{2} \widetilde U\right)\right) \right|.
\end{equation}
According to \cref{asption} and \cref{lemma:Hamiltonianflow}, the right hand side could be uniformly bounded for $[0,t]$ by a constant $\gamma(t)$ which doesn't depend on $N$ and $\varepsilon$. Then according to Gronwall's inequality, let $G(t)=\exp(\gamma(t))$, we have
$$
\frac{\mathrm{d}}{\mathrm{d} t}\left| A_{k,k_{[{\nu}/{2}]:1}}^{(\nu)}\right|
    \leq  \frac{1}{2} \left|A_{k,k_{[{\nu}/{2}]:1}}^{(\nu)}\right| \gamma(t) \Rightarrow \frac{\left|A_{k,k_{[{\nu}/{2}]:1}}^{(\nu)}\left(t ,T_{n:1}, z_{0}\right) \right|}{\left|A_0^{(0)}(0,z_0)\right|}\leq G(t).
$$
\end{proof}
\subsection{Discussions}
With the above analysis, we have established that to reach a certain accuracy, the sample size required by the FGS algorithm depends  neither on the number of metal orbitals $N$  nor on the semiclassical parameter $\varepsilon$. This is the most important advantage of the FGS algorithm. We will also numerically validate this property in \cref{sec:numerical}.

The numerical analysis of the surface hopping ansatz is inspired by the numerical analysis of one-level Frozen Gaussian Sampling \cite{xie2021frozen}. However, in one-level system, there is no surface hopping, therefore the trajectory becomes deterministic if its initial value has been specified. As a result, the error control of the propagating scheme is much simpler.

Our result, however, indicates that the hopping between energy surfaces doesn't induce additional sampling errors. To the best of our knowledge, this is the first rigorous mathematical analysis that prove this particular feature of surface hopping algorithms, which gives an explanation to why surface hopping methods have been found to be efficient even in the semi-classical limit.

Furthermore, we have proved the above result in the context of surface hopping at molecule-metal surface. A similar analysis could be applied to surface hopping on two band systems (or any finite band systems) \cite{lu2018frozen,lu2016improved}, where the authors have already observed numerically that the computational cost is essentially independent of the semiclassical parameter $\varepsilon$.

\section{Numerical Results}
\label{sec:numerical}
In this section, we present numerical experiments to verify the accuracy and convergence properties of our method. We  also numerically explore the Anderson-Holstein model to show that our method can capture the information of the quantum observables that classical trajectories  fail to capture.

In the following experiments, we   solve equation system \ref{eq:discretize_u} with the following initial value:
\begin{equation}
    \left\{
\begin{aligned}
&u_{0}(0, x)=(\pi\varepsilon)^{-1/4}\exp \left(-\frac{(x-q_0)^{2}}{2 \varepsilon}\right) \exp \left(\frac{ \mathrm i p_0(x-q_0)}{\varepsilon}\right),\\ 
&u_k(0,x)=0,\quad k=1,2,\cdots,N.
\end{aligned}\right.
\label{eq:numericalinitial}
\end{equation}
The initial condition corresponds to a wave packet  localized at $q=q_0$ with momentum $p=p_0$ in the molecule orbital.
We  numerically solve all the time-evolution ODEs \cref{eq:groupdynamics} using the fourth order Runge-Kutta scheme with the time step $\Delta t = 0.01$. Without loss of generality, we take $\mathcal{E}_a=0, \mathcal{E}_b=1$ unless specified otherwise.

\subsection{Accuracy and convergence test}
The goals of this subsection are two-fold: on the one hand, we aim  to verify the accuracy of the FGS algorithm, using the numerical solution obtained by time-splitting spectral method \cite{bao2002time} as the reference solution. On the other hand, we  demonstrate the efficiency of the FGS algorithm, by showing that the sample size required to reach a certain error threshold is independent of both the number of metal orbitals $N$ and the semi-classical parameter $\varepsilon$. Besides, the half order convergence of the FGS algorithm with respect to the sample size is observed.

Let us first consider the following example where $U_0$ and $U_1$ are harmonic potentials. This is the scenario that are  widely used in the modeling of nonadiabatc dynamics at metal surfaces \cite{cao2017lindblad,dou2015surface,dou2018perspective}.

\begin{figure}[htbp]
    \centering
    \includegraphics[width=1\textwidth]{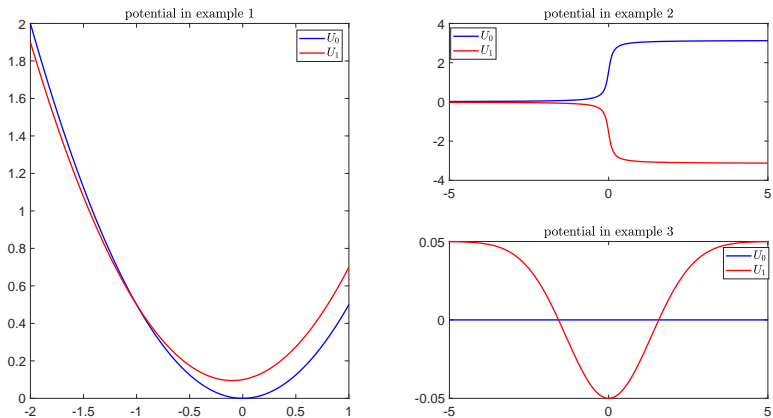}
    \caption{Plots of the potential functions $U_0$ and $U_1$ in Example 1, 2 and 3.}
    \label{fig:potential}
\end{figure}

\textbf{Example 1 (harmonic oscillators)}. Let 
$$
\begin{aligned}
& U_0(x)=\frac{x^2}{2},\quad   U_1(x) =\frac{x^2}{2}+0.1x+0.1,\\
&\quad q_0=-1.5, \quad p_0=2, \quad  V(\mathcal{E},x)=1.
\end{aligned}
$$

Let us first demonstrate the accuracy of the FGS algorithm. In  \cref{fig:example1_wavefunction}, with $M=20000$ sampled wave packets, we present the real part of the reconstructed wave functions $u_0(t,x)$ at the molecule orbital, and $u_1(t,x)$ at one of the metal orbitals, with $\varepsilon=\frac{1}{64}$, $N=16$ and $t=1$. We can see that the reconstructed wavefunction has a very high precision.

\begin{figure}[htbp]
    \centering
    \includegraphics[width=1\textwidth]{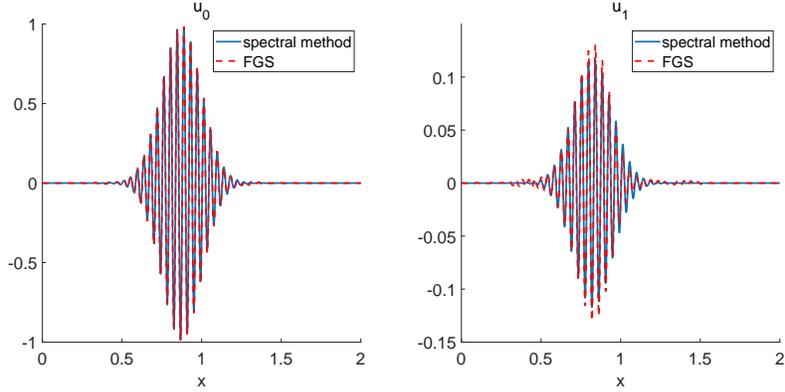}
    \caption{  (\textbf{Example 1}): The real part of wave functions $u_0(t,x)$ (left) and $u_1(t,x)$ (right) for $t=1$ with  $N=16$ and $\varepsilon=\frac{1}{64}$. The red line corresponds to FGS (sample size $M=20000$) and the blue line corresponds to the time-splitting spectral method.}
    \label{fig:example1_wavefunction}
\end{figure}

Now let us further investigate the sampling error of the FGS algorithm and numerically validate that it does not depend on $\varepsilon$ and $N$. 
We compare the $L^2$ sampling errors of FGS algorithm under two sets of scenarios: (1) fix a large enough number of orbitals $N=256$, let $\varepsilon = 1/32,1/64,1/128$; (2) fix a small enough semiclassical parameter $\varepsilon=\frac{1}{64}$, let $N=16,64,256$. The sampling errors of these two comparisons are listed in \cref{table:example1}. We can see that to reach the same precision, the number of wave packets that we need in the FGS algorithm is independent of both $\varepsilon$ and $N$.
Furthermore, a log-scale plot of the sampling error is available in \cref{fig:example1_order}, which verifies that the convergence order of our stochastic method is indeed $\frac{1}{2}$.

\begin{center}
  \begin{tabular}{c|ccc}
   \toprule
   \textbf{Example 1 ($N=256$)} & \(\varepsilon=1/32\) & \(\varepsilon=1/64\) & \(\varepsilon=1/128\)  \\
   \midrule
   \(M=3200\) & 1.64e-01 & 1.52e-01 & 1.55e-01  \\
   \(M=6400\) & 1.11e-01 & 1.15e-01 & 1.11e-01  \\
   \(M=12800\) & 8.00e-02 & 7.75e-02 & 7.59e-02  \\
   \(M=25600\) & 5.55e-02 & 5.52e-02 & 5.38e-02  \\
   \(M=51200\) & 4.08e-02 & 4.05e-02 & 4.00e-02   \\
   \midrule
   \textbf{Example 1 ($\varepsilon=1/64$)} & \(N=16\) & \(N=64\) & \(N=256\)  \\
   \midrule
   \(M=3200\) & 1.53e-01 & 1.59e-01 & 1.52e-01  \\
   \(M=6400\) & 1.10e-01 & 1.10e-01 & 1.15e-01  \\
   \(M=12800\) & 8.39e-02 & 8.11e-02 & 7.75e-02  \\
   \(M=25600\) & 5.39e-02 & 5.54e-02 & 5.52e-02  \\
   \(M=51200\) & 4.01e-02 & 4.00e-02 & 4.05e-02   \\
   \bottomrule
  \end{tabular}
 \captionof{table}{(\textbf{Example 1}) The $L^2$ errors of wave functions (at $t=1.5$) for  $\varepsilon=\frac{1}{32}, \frac{1}{64}, \frac{1}{128}$ versus various sample size $M$ with fixed $N=256$ (top) and for $N=16, 64, 256$ versus various sample size $M$ with fixed $\varepsilon=\frac{1}{64}$ (bottom). }
 \label{table:example1}
 \end{center}

\begin{figure}[htbp]
    \centering
    \includegraphics[width=1\textwidth]{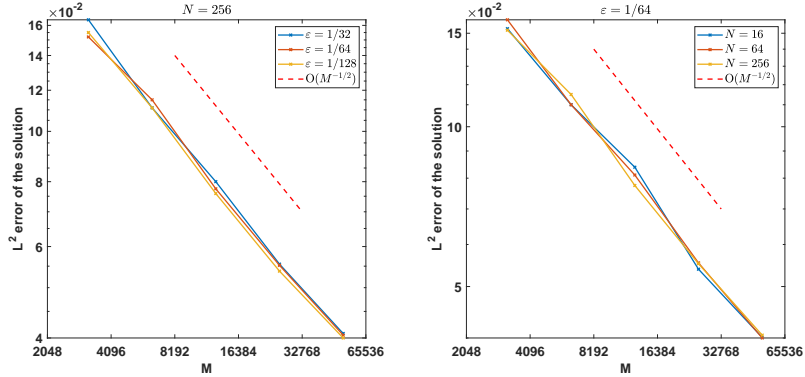}
    \caption{(\textbf{Example 1}) The $L^2$ error versus the sample size $M$, for fixed $N=256$ with various $\epsilon=\frac{1}{32},\frac{1}{64},\frac{1}{128}$ (left) and for  fixed $\varepsilon=\frac{1}{64}$ with various $N=16,64,256$ (right). 
    }
    \label{fig:example1_order} 
\end{figure}

Next, we move on to more complicated examples. Let us consider the extended coupling with reflection  (example 2) and the inhomogeneous transition (example 3). The extended coupling with reflection has been one of the most challenging test cases for surface hopping algorithm  since Tully's pioneering work \cite{tully1990molecular}, while the inhomogeneous transition is widely used to incorporate generic coupling between molecule and metal orbitals \cite{dou2018perspective,dou2015surface}.

\textbf{Example 2 (extended coupling with reflection)} Let
\begin{gather*}
 U_0(x)=\arctan(10x)+\frac{\pi}{2},\quad  U_1(x)=-U_0(x),\\
q_0=-1.5, \quad p_0=2, \quad V(\mathcal{E},x)=1.
\end{gather*}

\textbf{Example 3 (inhomogeneous
 transition)} Let
\begin{gather*}
U_{0}(x)=0, \quad U_{1}(x)=-0.1 \exp \left(-0.28 x^{2}\right)+0.05,\\
 q_0=-2.5,\quad p_0=2 ,\quad V(\mathcal E,x)=\exp \left(-0.06 x^{2}\right).
\end{gather*}

\begin{figure}[htbp]
    \centering
    \includegraphics[width=\textwidth]{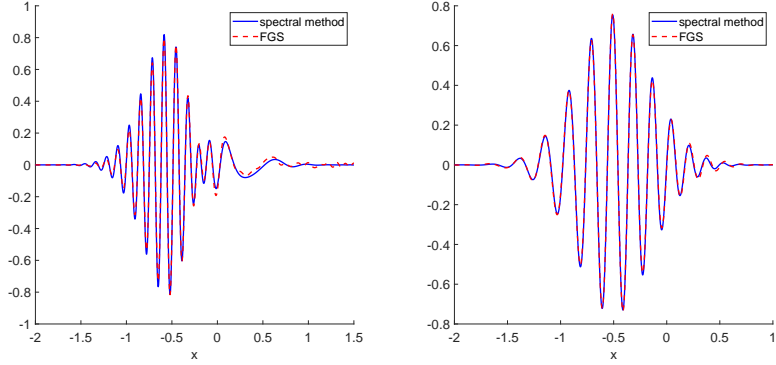}
    \caption{Left (Example 2): The real part of  $u_0(t,x)$ at $t=1.4$ with $\varepsilon=0.04, N=5, M=16000$. Right ({Example 3}): The real part of $u_0(t,x)$ at $t=1$ with $\varepsilon=1/16, N=5, M=10000$.}
    \label{fig:exampleu0}
\end{figure}

The real part of the reconstructed wave function $u_0(t,x)$ of these two examples is presented in \cref{fig:exampleu0}. In example 2, the wave packet is initialized with the center placed at $q_0=-1.5$ and  is set to travel to the right, and gets reflected back to the left. In example 3, the jumping  intensity peaks around the origin and decays rapidly away from the center. These features make the computation much more difficult than that in example 1. Nonetheless, we can still see that the reconstructed wave functions in \cref{fig:exampleu0} have  pretty high precision.

Let us  examine the convergence property  of FGS algorithm once again using example 2, similar as what we did in example 1. The comparison of sampling errors is listed in \cref{table:example2}. We can still see that the sampling error is independent of both $N$ and $\varepsilon$. A log-scale plot of the sampling error \cref{fig:example2_order} demonstrates its half-order convergence with respect to the number of sampled wavepackets $M$.

\begin{center}
  \begin{tabular}{c|ccc}
   \toprule
   \textbf{Example 2 ($N=256$)} & \(\varepsilon=1/32\) & \(\varepsilon=1/64\) & \(\varepsilon=1/128\)  \\
   \midrule
   \(M=3200\) & 2.70e00 & 7.50e-01 & 9.82e-01  \\
   \(M=6400\) & 2.20e-01 & 1.96e-01 & 1.87e-01  \\
   \(M=12800\) & 1.38e-01 & 1.35e-01 & 1.27e-01  \\
   \(M=25600\) & 1.00e-01 & 9.79e-02 & 8.92e-02  \\
   \(M=51200\) & 7.21e-02 & 6.81e-02 & 6.41e-02   \\
   \midrule
   \textbf{Example 2 ($\varepsilon=1/64$)} & \(N=16\) & \(N=64\) & \(N=256\)  \\
   \midrule
   \(M=3200\) & 2.66e-01 & 2.68e-01 & 7.50e-01  \\
   \(M=6400\) & 1.92e-01 & 1.93e-01 & 1.96e-01  \\
   \(M=12800\) & 1.39e-01 & 1.37e-01 & 1.34e-01  \\
   \(M=25600\) & 9.46e-02 & 9.47e-02 & 9.79e-02  \\
   \(M=51200\) & 6.73e-02 & 6.72e-02 & 6.81e-02   \\
   \bottomrule
  \end{tabular}
 \captionof{table}{\textbf{(Example 2) }The  $L^2$ error (at $t=1.5$)  for various $\varepsilon=\frac{1}{32},\frac{1}{64},\frac{1}{128}$ versus various sample size $M$ with fixed $N=256$ (top) and for various $N=16,64,256$ versus various sample size $M$ with fixed $\varepsilon=\frac{1}{64}$ (bottom) .}
 \label{table:example2}
 \end{center}

\begin{figure}[htbp]
    \centering
    \includegraphics[width=1\textwidth]{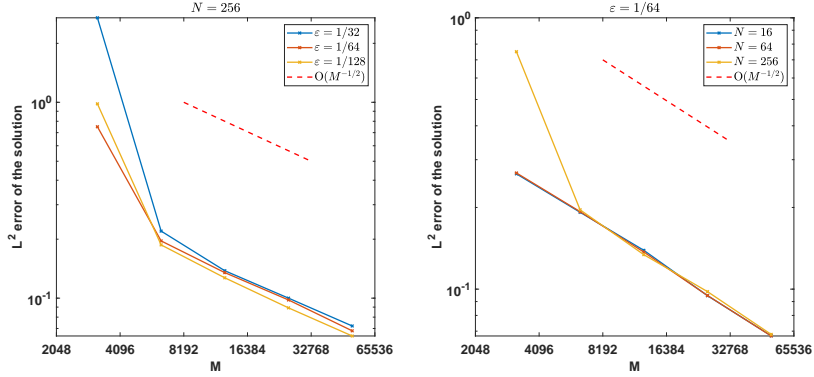}
    \caption{(\textbf{Example 2}) The $L^2$ error  versus the sample size $M$, for fixed $N=256$ with various $\epsilon=\frac{1}{32},\frac{1}{64},\frac{1}{128}$ (left) and for  fixed $\varepsilon=\frac{1}{64}$ with various $N=16,64,256$ (right). 
    }
    \label{fig:example2_order}
\end{figure}

\subsection{Numerical explorations}
With the FGS algorithm, we are now ready to explore further on the nonadiabatic dynamics at metal surfaces. On the one hand, we would like to emphasize that
our method is  superior to the classical trajectories surface hopping methods by showing that the quantum observables of the nonadiabatic dynamics can not be fully captured by the ensemble average of classical trajectories. On the other hand, we  use FGS algorithm to explore the physics of metal surfaces, such as the transition rates with different energy gaps and the finite temperature effect.

\subsubsection{Quantum observables v.s. classical ensemble average}
Let $\hat O \in L^2(\mathbb{R}^m,\mathbb{C}^{N+1})$ be the quantum observable associated with the classical quantity $O(q,p)$. Based on the FGS algorithm, the expectation value of the quantum observable is
\begin{equation}
    \langle \hat O\rangle_{\text{q}}(t;\varepsilon)=\int \overline{u}_0(t,x)\hat O u_0(t,x)+\frac{1}{N}\left(\sum_{k=1}^N\overline{u}_k(t,x)\hat O u_k(t,x)\right)\mathrm{d}x,
    \label{eq:quantumobservable}
\end{equation}
However, using the classical trajectories sampled from $\widetilde z_t$ \cref{eq:Hanmiltonianflow}, we can also calculate the classical ensemble average of $O(q,p)$:
\begin{equation}
    \langle  O\rangle_{\text{c}}(t)=\int  O(q,p)\pi_t(q,p) \mathrm{d}q\mathrm{d}p,
    \label{eq:classicalobservable}
\end{equation}
where $\pi_t(q,p)$ is the distribution at time $t$ determined by the initial condition $\pi(q,p)$ and the stochastic evolution \cref{eq:Hanmiltonianflow}. In practice, we have
\begin{equation}
    \langle  O\rangle_{\text{c}}(t)\approx \frac{ 1}{M}\sum_{i=1}^M O(q_i,p_i),
\end{equation}

Note that in the FGS algorithm, the classical trajectory does not depend on $\varepsilon$. In other words, the classical ensemble average  is independent of $\varepsilon$, while the expectation valquantum observable definitely depends on $\varepsilon$. One might expect that when $\varepsilon\rightarrow 0$, $\langle \hat O\rangle_{\text{q}}(t;\varepsilon)$ would converge to $\langle  O\rangle_{\text{c}}(t)$, which is the  motivation of many surface hopping methods. 
However, we demonstrate below that this is not true.

Let us consider $\hat O$ to be the nuclear position $\hat x$. Following \cref{eq:quantumobservable} and \cref{eq:classicalobservable}, we can calculate the quantum expectation $\langle\hat x\rangle_{\text{q}}(t;\varepsilon)$ ($\varepsilon$ is taken to be $1/2$, $1/4$, $1/8$, $1/16$, $1/32$ and $1/64$) and the classical average $\langle x\rangle_{\text{c}}(t)$ in example 1. The result is shown in \cref{fig:quantum and classical center}.

\begin{figure}[htbp]
    \centering
    \includegraphics[width=1\textwidth]{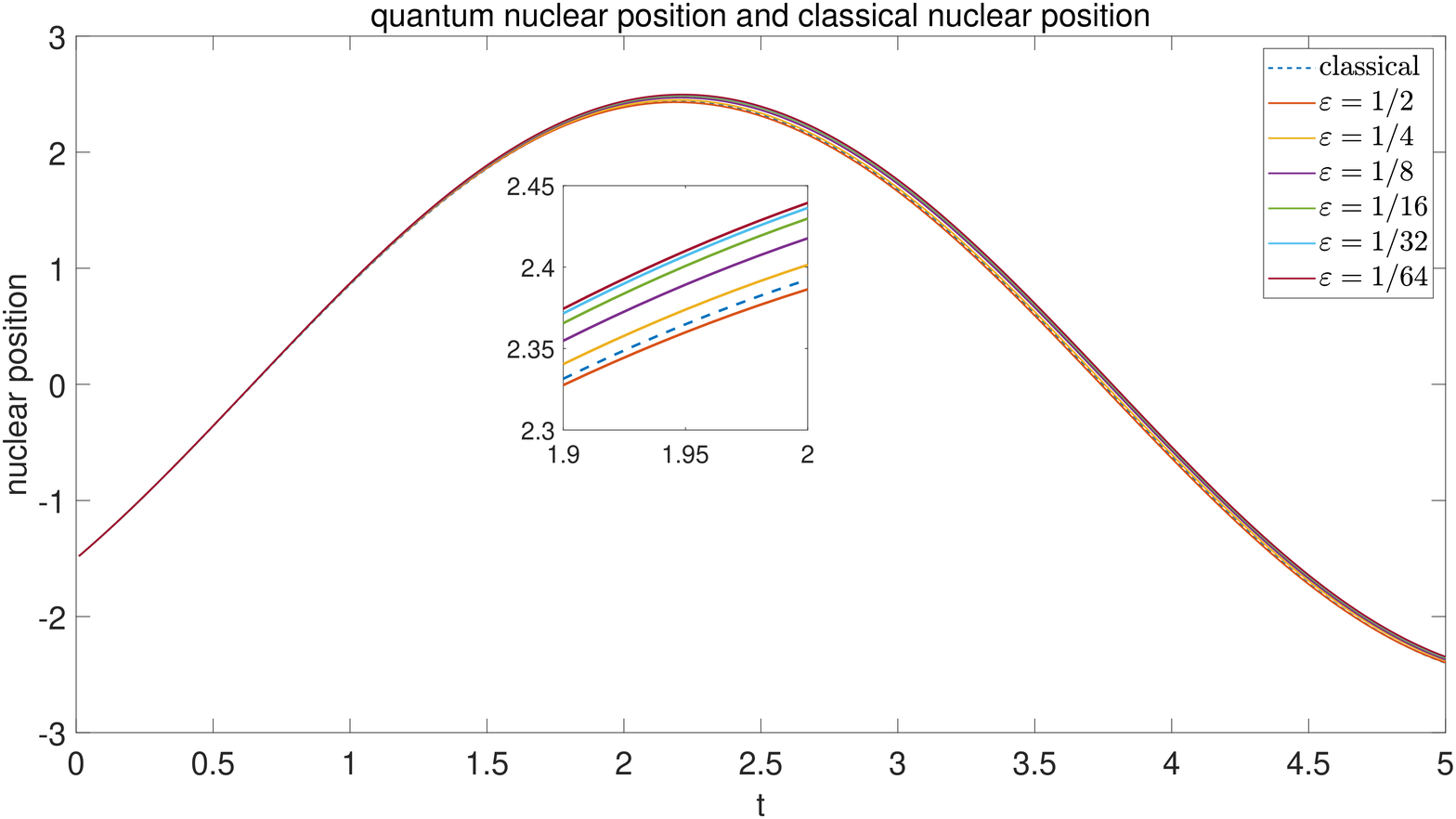}
    \caption{The comparison of classical average $\langle x\rangle_{\text{c}}(t)$ (dash line) and quantum expectation $\langle\hat x\rangle_{\text{q}}(t;\varepsilon)$ (solid line) ($\varepsilon$ = $1/2$, $1/4$, $1/8$, $1/16$, $1/32$ and $1/64$) of the nuclear position in example 1. One can see that as $\varepsilon\rightarrow 0$, $\langle\hat x\rangle_{\text{q}}(t;\varepsilon)$ does not converge to $\langle x\rangle_{\text{c}}(t)$.}
    \label{fig:quantum and classical center}
\end{figure}

One can clearly see that as $\varepsilon\rightarrow 0$, $\langle\hat x\rangle_{\text{q}}(t;\varepsilon)$ does not converge to $\langle x\rangle_{\text{c}}(t)$. This means that the classical trajectories can not capture all the physics of the system and one must rely on a quantum-level method, such as FGS algorithm to capture this information.
 Formally speaking, the FGS method tracks the phase information of each wave packet, and the superposition of those wave packets yields a subtle cancellation which the ensemble of classical trajectories can not provide.

\subsubsection{Electron transfer and finite temperature effect} 
Now, let us numerically explore the physics of Anderson-Holstein model with more specified setting.
Let us first simulate the electron transfer at metal surfaces using example $3$. The transition rate at $t=0.5$, calculated by both FGS and the time-splitting scheme,  are shown in \cref{fig:transition_rate}, with various $\mathcal{E}_a=0,\frac{1}{32},\frac{1}{16},\frac{1}{8},\frac{1}{4}$ but fixing $\mathcal{E}_b-\mathcal{E}_a=1$. We remark that $\mathcal{E}_a$ can be interpreted as the energy gap between the molecule orbital and the metal orbitals.  As $\mathcal{E}_a$ decreases, the energy gap gets smaller, and then it's easier for electrons to transfer between energy surfaces, and therefore the system has a higher hopping rate. This is indeed reflected in \cref{fig:transition_rate}. 
\begin{figure}[htbp]
    \centering
    \includegraphics[width=1\textwidth]{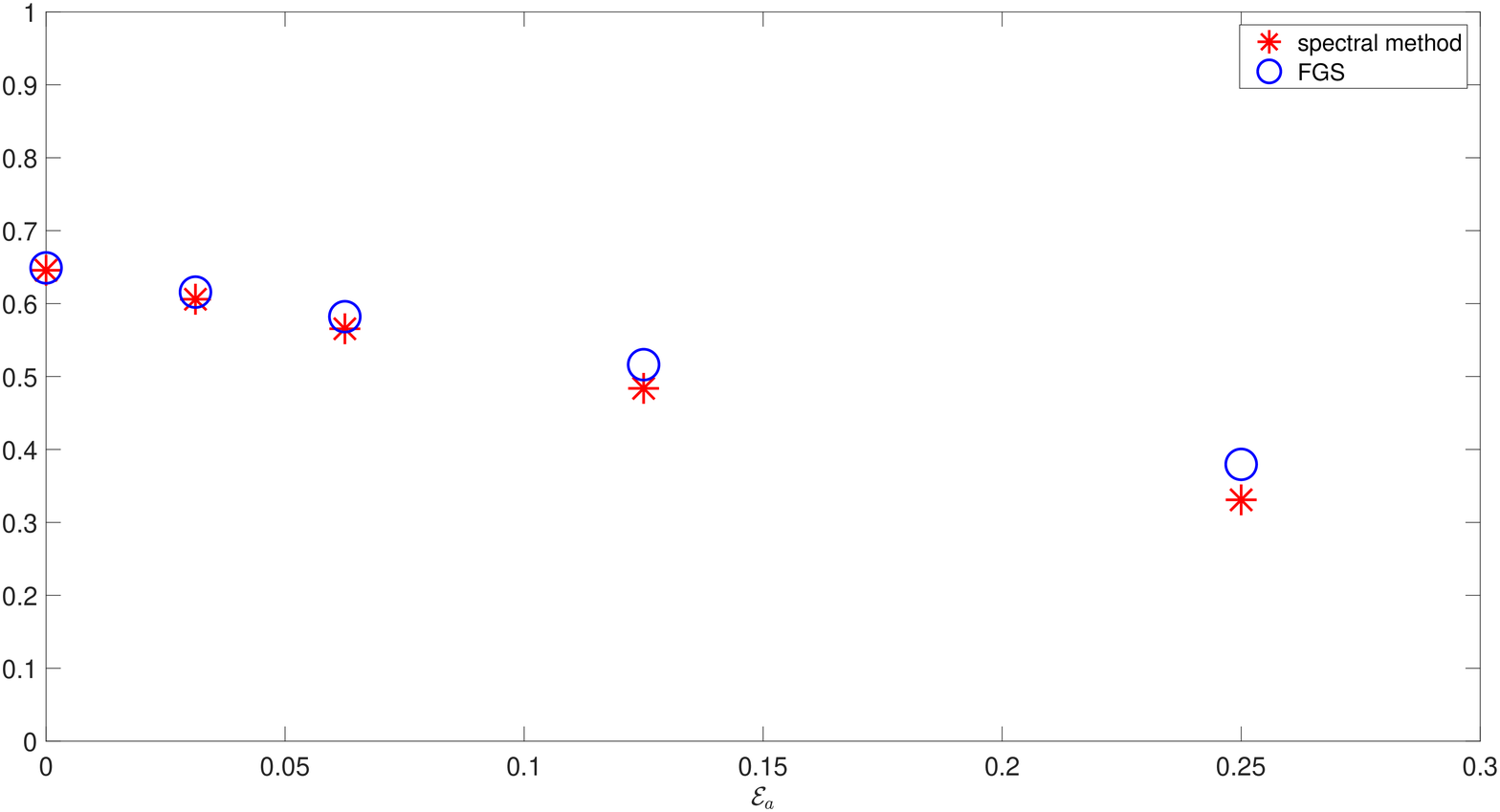}
    \caption{Transition rates at $t=0.5$ in example $3$ with $\varepsilon=\frac{1}{16}, N=16, V(\mathcal{E},x)=5$, sample size $M=120000$ and various $\mathcal{E}_a=0,\frac{1}{32},\frac{1}{16},\frac{1}{8},\frac{1}{4}$, the blue circle corresponds to FGS and the red asterisk corresponds to the time-splitting spectral method.}
    \label{fig:transition_rate}
\end{figure}

However, according to \cref{eq:Hanmiltonianflow}, we can see that  different metal orbital energies do not affect  the classical trajectories,  but only contribute to the quantum phases in the wave function. 
In other words, the change of transition rate with different energy band gaps are fundamentally rooted in
the cancellation of the phase of the wave packet. This is also a truly quantum effect that can not be captured by classical trajectories.

We are also able to efficiently simulate the finite temperature effect at metal surfaces \cite{dou2018perspective}, by letting the coupling potential $V(\mathcal E,x)$  be a Fermi-Dirac distribution:

\textbf{Example 4.}(Finite temperature effect) 
$$
\begin{aligned}
&U_{0}(x)=0, \quad U_{1}(x)=-0.1 \exp \left(-0.28 x^{2}\right)+0.05,\\
& \quad p_0=2,\quad q_0=-2.5, \quad  V(\mathcal{E},x)=\frac{8}{1+\exp \left(\beta\mathcal{E}\right)}.
\end{aligned}
$$
Here $\beta$ is the inverse temperature. For different values of $\beta$  ($\beta = 1,16$), we plot the real part of the reconstructed wave functions $u_0(t,x)$ of the molecule orbital at $t=0.5$ (with $\varepsilon=1/16$, $N=16$) in \cref{fig:Fermi_Dirac.u0}.
We can see that the FGS algorithm can efficiently simulate both low and high temperature scenarios at metal surfaces.

\begin{figure}[htbp]
    \centering
    \includegraphics[width=0.95\textwidth]{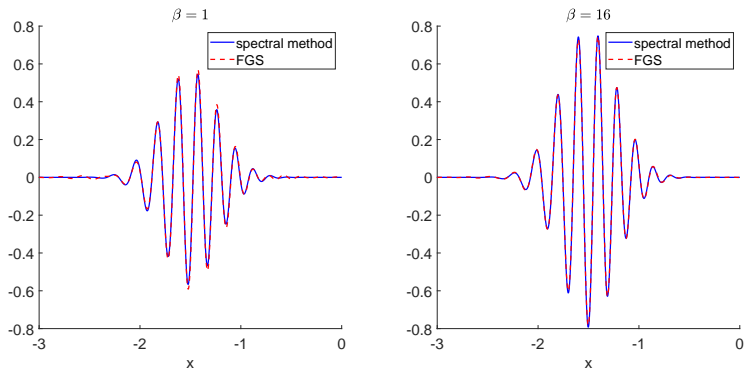}
    \caption{The real part of wave function $u_0(t,x)$ for $t=0.5$ in example $4$ with sample number $M=50000$, $\varepsilon=\frac{1}{16},N=16,\beta=1$ (left) and $\beta=16$ (right), the red line corresponds to FGA and the blue line corresponds to time-splitting spectral method.}
    \label{fig:Fermi_Dirac.u0}
\end{figure}
From a physical point of view, a higher temperature (i.e. a smaller $\beta$)  induces more  electron transfer. This is indeed observed in \cref{fig:Fermi_Dirac.u0}:
 a smaller $\beta$ corresponds to a  larger amplitude of the wave function at the molecule orbital, which indicates that there is less electron transfer.

\section{Conclusions and Discussions}
In this work, we have developed a Frozen Gaussian Sampling (FGS) method for efficiently simulating nonadiabatic quantum dynamics at metal surfaces. The key advantage is that its computational cost is essentially independent of both the semi-classical parameter $\varepsilon$ and the number of metal orbitals $N$. To our best knowledge, we have provided  the first rigorous mathematical proof that surface hopping strategies do not induce additional sampling errors.

The FGS algorithm provides various new possibilities for studying nonadiabatic dynamics at metal surfaces. 
For example, we are modeling the metal surfaces as a closed quantum system in this work. However, the metal surfaces are often  viewed as an open quantum system coupled with multiple baths subject to the Boltzmann distribution \cite{dou2015surface,dou2018perspective}.
The strategy of the FGS could, in theory, be generalized towards the open system setting.
What's more, nonadiabatic dynamics in the strong-coupling regime is  very challenging to simulate. This has been discussed in \cite{cai2022asymptotic} for the two-level spin-boson model, but it is still unclear whether there exists a generic approach for all strong-coupling models. These are left for future research.

\appendix

\section{Derivation of the method}
\label{sec:derivation}
Recall that
$$
\left\{
\begin{aligned}
\mathrm{i}\varepsilon \partial_{t} u_{0}(t, x)&=-\frac{\varepsilon^{2}}{2} \Delta u_{0}(t, x)+U_{0}(x) u_{0}(t, x)+\varepsilon \sum_{k=1}^Nh V\left({\mathcal{E}_k}, x  \right)u_{k}(t, x),\\
\mathrm{i}\varepsilon \partial_{t} u_{k}(t, x)&=-\frac{\varepsilon^{2}}{2} \Delta u_{k}(t, x)+\left(U_{1}(x)+{\mathcal{E}}_k\right) u_{k}(t, x)+\varepsilon \overline{V}({\mathcal{E}_k}, x) u_{0}(t, x),\text{ } k=1:N.
\end{aligned}\right.
$$
and our ansatz
$$
\begin{aligned}
{u}_{0,\text{FG}}(t,x)&=  u_0^{(0)}(t,x)+ \sum_{k_1=1}^N u_{0,k_1}^{(2)}(t,x)+\cdots,\\
{u}_{k,\text{FG}}(t,x)&= \sum_{k=1}^N u_{k}^{(1)}(t,x)+ \sum_{k_1=1}^N u_{k,k_1}^{(3)}(t,x)+\cdots,\quad k=1,\cdots,N.
\end{aligned}
$$
The ansatz for $u_{0}^{(0)}(t,x)$ is
\begin{equation}
 u_0^{(0)}(t, x)=\frac{1}{(2 \pi \varepsilon)^{3 m / 2}} \int A_0^{(0)}\left(t, z_{0}\right) \exp \left(\frac{\mathrm i}{\varepsilon} \Theta_0^{(0)}\left(t, z_{0}, x\right)\right) \mathrm{d} z_{0}.
\end{equation}
The ansatz for $u_{k}^{(1)}(t,x)$ is
\begin{equation}
\begin{aligned}
 u_{k}^{(1)}(t, x)=& \frac{1}{(2 \pi \varepsilon)^{3 m / 2}}  \int \mathrm{d} z_{0} \int_{0}^{t} \mathrm{d} t_{1}\\
   &\tau_{k}^{(1)}\left(t_1, z_{0}\right)\times 
 A_{k}^{(1)}\left(t,t_1, z_{0}\right) \exp \left(\frac{\mathrm i}{\varepsilon} \Theta_{k}^{(1)}\left(t, t_1, z_{0}, x\right)
\right).
\end{aligned}
\end{equation}
We have 
\begin{equation}
\begin{aligned}
&\mathrm{i}\varepsilon \partial_{t} u_{k}^{(1)}(t, x)+\frac{\varepsilon^{2}}{2} \Delta u_{k}^{(1)}(t, x)-\left(U_{1}(x)+{\mathcal{E}}_k\right) u_{k}^{(1)}(t, x)\\
=& \frac{\mathrm{i}\varepsilon}{(2 \pi \varepsilon)^{3 m / 2}}  \int \mathrm{d} z_{0} \tau_{k}^{(1)}\left(t, z_{0}\right)\times 
 A_{k}^{(1)}\left(t,t, z_{0}\right) \exp \left(\frac{\mathrm i}{\varepsilon} \Theta_{k}^{(1)}\left(t, t, z_{0}, x\right)
\right).
\end{aligned}
\end{equation}
With comparison to the equation for $u_k$, we actually want to match the following term:
\begin{equation}
\begin{aligned}
&\mathrm{i}\varepsilon \partial_{t} u_{k}^{(1)}(t, x)+\frac{\varepsilon^{2}}{2} \Delta u_{k}^{(1)}(t, x)-\left(U_{1}(x)+{\mathcal{E}}_k\right) u_{k}^{(1)}(t, x)=
\varepsilon \overline{V}({\mathcal{E}_k}, x) u_{0}^{(0)}(t, x),
\end{aligned}
\end{equation}
i.e.
\begin{equation}
\begin{aligned}
&\mathrm{i}\int \mathrm{d} z_{0} \tau_{k}^{(1)}\left(t, z_{0}\right)\times 
 A_{k}^{(1)}\left(t,t, z_{0}\right) \exp \left(\frac{\mathrm i}{\varepsilon} \Theta_{k}^{(1)}\left(t, t, z_{0}, x\right)
\right) \\
= &\overline{V}({\mathcal{E}_k}, x)\int A_0^{(0)}\left(t, z_{0}\right) \exp \left(\frac{\mathrm i}{\varepsilon} \Theta_0^{(0)}\left(t, z_{0}, x\right)\right) \mathrm{d} z_{0}.
\end{aligned}
\end{equation}
Therefore, we let
$$
A_{k}^{(1)}\left(t,t, z_{0}\right) = A_0^{(0)}\left(t, z_{0}\right),\quad \Theta_{k}^{(1)}\left(t,t, z_{0}\right) = \Theta_0^{(0)}\left(t, z_{0}\right),
$$
and we let
$$
\tau_{k}^{(1)}\left(t, z_{0}\right) =-\mathrm{i} \overline{V}({\mathcal{E}_k}, Q_0^{(0)}),
$$
where we use $\overline{V}({\mathcal{E}_k}, Q_0^{(0)})$ to approximate $\overline{V}({\mathcal{E}_k}, x)$. This approximation could be justified using Taylor's expansion, and it is rigorously proved in \cite{lu2018frozen} that it only introduces an $O(\varepsilon)$ error.

Now let's turn to the ansatz for
\begin{equation}
\begin{aligned}
 u_{0,k_{1}}^{(2)}&(t, x)=  \frac{1}{(2 \pi \varepsilon)^{3 m / 2}} \int \mathrm{d} z_{0} \int_{0}^{t} \mathrm{~d} t_{2} \int_{0}^{t_{2}} \mathrm{~d} t_{1}\\
  &h\tau_{k_1}^{(1)}\left(t_1, z_{0}\right)\tau_{{k_1}}^{(2)}\left(t_2,t_1, z_{0}\right)
A_{0,k_{1}}^{(2)}\left(t, t_2,t_1, z_{0}\right) \exp \left(\frac{\mathrm i}{\varepsilon} \Theta_{0,k_{1}}^{(2)}\left(t, t_2,t_1, z_{0}, x\right).
\right) 
\end{aligned}
\end{equation}
Plugging in this ansatz into the equation for $u_0$, we identify the term needs to be matched:
\begin{equation}
\begin{aligned}
&  \mathrm{i}\varepsilon\int \mathrm{d} z_{0}  \int_{0}^{t} \mathrm{~d} t_{1}h\tau_{k_1}^{(1)}\left(t_1, z_{0}\right)\tau_{{k_1}}^{(2)}\left(t,t_1, z_{0}\right)
A_{0,k_{1}}^{(2)}\left(t, t,t_1, z_{0}\right) \exp \left(\frac{\mathrm i}{\varepsilon} \Theta_{0,k_{1}}^{(2)}\left(t, t,t_1, z_{0}, x\right)
\right) \\
 &=\varepsilon h V\left({\mathcal{E}_{k_1}}, x  \right)\int \mathrm{d} z_{0} \int_{0}^{t} \mathrm{d} t_{1}\tau_{{k_1}}^{(1)}\left(t_1, z_{0}\right)\times 
 A_{k_1}^{(1)}\left(t,t_1, z_{0}\right) \exp \left(\frac{\mathrm i}{\varepsilon} \Theta_{k_1}^{(1)}\left(t, t_1, z_{0}, x\right)\right).
\end{aligned}
\end{equation}
Therefore, we let
$$
A_{0,k_{1}}^{(2)}\left(t, t,t_1, z_{0}\right) = A_{k_1}^{(1)}\left(t,t_1, z_{0}\right),\quad \Theta_{0,k_{1}}^{(2)}\left(t, t,t_1, z_{0}\right) = \Theta_{k_1}^{(1)}\left(t,t_1, z_{0}\right),
$$
and we let
$$
\tau_{k_1}^{(2)}\left(t, z_{0}\right) =-\mathrm{i} {V}({\mathcal{E}_k}, Q_{k_1}^{(1)}).
$$
where we use ${V}({\mathcal{E}_k}, Q_{k_1}^{(1)})$ to approximate ${V}({\mathcal{E}_k}, x)$. This approximation again only introduces an $O(\varepsilon)$ error.

Carrying on this process, we arrive at the dynamics \cref{eq:groupdynamics}, the initial value \cref{eq:initialQPS,eq:dynamicsforamplitude2,eq:initialQPSA2}, and the hopping coefficient \cref{eq:hoppingcoefficient}. Its mathematical justification could be done by the exact same semiclassical analysis in \cite{lu2018frozen}, of which we omit the details here.

\section*{Acknowledgement}
The work of Z.Z. was  supported by the National Key R\&D Program of China, Project Number 2020YFA0712000, 2021YFA1001200 and NSFC grant Number 12031013, 12171013. The work of Z.H. was partially supported  by the National Science Foundation under Grant No. DMS-1652330. Z.Z. thanks Dr. Wenjie Dou for helpful discussions.
 L.X. thanks Dr. Hao Wu for his support and encouragement.

\bibliographystyle{elsarticle-num}
\bibliography{SH_metal}

\end{document}